\documentclass[11pt]{article}
\usepackage[margin=1in]{geometry}
\usepackage{amsmath}
\usepackage{amsthm}
\usepackage{bbm}
\usepackage{bm}
\usepackage{amssymb}
\usepackage{setspace}
\usepackage{color}
\usepackage{float}
\usepackage{dsfont}
\usepackage{enumerate}
\usepackage{hyperref}
\usepackage[capitalize]{cleveref}
\usepackage{graphicx}
\usepackage{algpseudocode,algorithm,algorithmicx}
\usepackage{tikz}
\usepackage{subcaption} 
\usepackage{wrapfig}
\usepackage{makecell}
\usepackage{mathdots}
\usepackage{cancel}
\usepackage[utf8]{inputenc}
\usepackage[normalem]{ulem}
\usepackage{tikz}
\newtheorem{theorem}{Theorem}
\newtheorem{proposition}[theorem]{Proposition}
\newtheorem{lemma}[theorem]{Lemma}

\newtheorem{corollary}[theorem]{Corollary}

\theoremstyle{definition}

\newtheorem{remark}[theorem]{Remark}

\numberwithin{theorem}{section} 

\newcommand{\nwc}{\newcommand*}
\nwc{\beq}{\begin{eqnarray}}
\nwc{\eeq}{\end{eqnarray}}

\numberwithin{equation}{section}

\usepackage{setspace}
\def\C{\mathbb{C}}

\def\R{\mathbb{R}}

\def\T{\mathbb{T}}

\def\Z{\mathbb{Z}}

\renewcommand{\subset}{\subseteq}

\renewcommand{\tilde}{\widetilde}
\renewcommand{\epsilon}{\varepsilon}
\renewcommand{\Re}{\text{Re}}

\def\range{{\rm range}}

\newcommand{\commentout}[1]{}
\newcommand{\bfa}{{\boldsymbol a}}
\newcommand{\bfb}{{\boldsymbol b}}

\newcommand{\bfe}{{\boldsymbol e}}

\newcommand{\bft}{{\boldsymbol t}}
\newcommand{\bfu}{{\boldsymbol u}}

\newcommand{\bfx}{{\boldsymbol x}}
\newcommand{\bfy}{{\boldsymbol y}}
\newcommand{\bfz}{{\boldsymbol z}}

\newcommand{\bfzero}{{\boldsymbol 0}}

\newcommand{\bfphi}{{\boldsymbol \phi}}

\newcommand{\calE}{\mathcal{E}}

\newcommand{\calH}{\mathcal{H}}

\newcommand{\calJ}{\mathcal{J}}

\newcommand{\calP}{\mathcal{P}}
\newcommand{\calT}{\mathcal{T}}

\newcommand{\Gr}{\mathrm{Gr}}

\newcommand{\diag}{{\rm diag}}

\newcommand{\poly}{{\rm poly}}

\def\wherespace{\quad\text{where}\quad}

\def\forallspace{\quad\text{for all}\quad}

\def\andspace{\quad\text{and}\quad}

\newcommand{\norm}[1]{{\big\|#1\big\|}}
\newcommand{\round}[1]{{\left(#1\right)}}

\numberwithin{equation}{section}

\begin{document}
	
	\title{Optimal Structured Approximation of Fourier Subspaces, Toeplitz Matrices, and Exponential Sums}
	
	\author{Albert Fannjiang\thanks{University of California, Davis. Email: cafannjiang@ucdavis.edu} \and Weilin Li\thanks{City University of New York, City College and Graduate Center. Email: wli6@ccny.cuny.edu}}

	\maketitle

	\begin{abstract}
		This paper studies three structured approximation problems: (1) Recovering the range of a Fourier matrix from a single observation, (2) Recovering a corrupted low-rank Toeplitz/Hankel matrix, and (3) Recovering a finite exponential sum from noisy samples. All three problems are computationally challenging because their structural constraints are difficult to enforce directly. We show that all three tasks can be solved efficiently and optimally by applying the Gradient-MUSIC algorithm for spectral estimation. To provide an example, for a rank-$r$ Toeplitz matrix $T\in {\mathbb C}^{n\times n}$ that satisfies a regularity assumption and is corrupted by an arbitrary $E\in {\mathbb C}^{n\times n}$ such that $\|E\|_2\leq \alpha n$, our algorithm outputs a Toeplitz matrix $T_\sharp$ of rank exactly $r$ such that $\|T-T_\sharp\|_2 \leq C \|E\|_2$, where $C,\alpha>0$ are absolute constants. This performance guarantee is minimax optimal in $n$, $r$, and $\|E\|_2$. For the other two structured approximation problems, we also provide algorithms that are minimax optimal in the number of samples, rank/sparsity, and noise level. At the heart of this paper is a quantitative transference principle which shows how to convert computational methods and theory for spectral estimation into corresponding methods and theory for the other three problems. 
	\end{abstract}
	
	
	
	% % % % % % % % % % % % % % % % % % % 
	
	\setcounter{tocdepth}{1}
	\tableofcontents

		\section{Introduction}

Sampling and approximation theory asks how an object (e.g., function, matrix, etc.) may be reconstructed from finitely many pieces of information (e.g., measurements, samples). A basic example is the exponential sum $f\colon\R\to\C$,
\[
f(t)=\sum_{k=1}^r a_k e^{itx_k},
\qquad t\in\mathbb R,
\]
where each frequency \(x_k\in\mathbb T:=\mathbb R/2\pi\mathbb Z\) lies on the continuum, each
amplitude \(a_k\in\mathbb C\) is nonzero.  If $f$ is sampled at integer times \(t=0,\dots,n-1\) and perturbed by additive noise $\bfz\in \C^n$, the model becomes
\[
\tilde \bfy:=\bfy+\bfz, \qquad 
\bfy:=\Phi(\bfx)\bfa, \qquad
\Phi(\bfx)
:= \Phi_n(\bfx)
:=
\bigl[e^{ijx_k}\bigr]_{j=0,\dots,n-1;\,k=1,\dots,r}.        
\]
By a scaling argument, without loss of generality, we assume that $|a_j|\geq 1$ for all $j$. 

The {\it spectral estimation problem} is to recover the frequencies $\bfx$ from samples $\tilde\bfy$. Unlike the setting of Fourier series, the unknown function is not bandlimited on a fixed frequency grid and the spectral nodes themselves are off-grid and must be inferred from the data, which makes approximation of $\bfx$ inherently a nonlinear inverse problem. The minimum separation
\[
\Delta(\bfx)=\min_{j\neq k}|x_j-x_k|_{\mathbb T}
\]
governs the stability of this inverse problem. Roughly speaking, we say $\bfx$ is {\it well-separated} if for a big enough $\beta\geq 1$, it holds that
\[
\Delta(\bfx)\geq \frac {2\pi \beta} n.
\]
In which case, the {\it Fourier matrix} $\Phi_n(\bfx)$ has a uniformly bounded condition number \cite{aubel2019vandermonde} and minimax optimal estimation of $\bfx$ is possible through the Gradient-MUSIC algorithm \cite{fannjiang2025optimality}.

In this paper, we use spectral estimation as the core problem and Gradient-MUSIC as the algorithm of our choice, in order to show how the following three structured approximation problems can be solved in a minimax optimal manner. Let us mention right now that this paper is organized in such a way that one may understand any of the following three problems independent from the remaining ones. 

\begin{enumerate}
	\item 
A {\it Fourier subspace} is a subspace of the form
\[
U(\bfx):=\range \big(\Phi_n(\bfx)\big) \subset\C^n.
\]
%When $n\geq 2r$, there is a bijection between all subsets of $\T$ of cardinality $r$ and Fourier subspaces of dimension $r$ in $\C^n$. 
Given measurement $\tilde\bfy$, {\it Fourier subspace estimation} asks to find a Fourier subspace $\tilde U\subset \C^n$ of dimension $r$ that approximates $U$ (with respect to either the spectral  or Frobenius distance). This problem is only nontrivial due to the structural requirement that $\tilde U$ is a Fourier subspace. 

Many procedures in data analysis and signal processing depend primarily on a leading subspace rather than the full matrix. This is true, for example, in classical subspace methods for array signal processing \cite{krim1996two}, in principal component analysis, and in many spectral
methods in data science \cite{chen2021spectral}. 

\item 
Consider a rank-$r$ Toeplitz matrix that enjoys the factorization
$$
T=\Phi_n(\bfx)\diag(\bfa)\Phi_n(\bfx)^*\in \C^{n\times n}.
$$
If $E$ is an arbitrary perturbation, given
\[
         M:=T+ E\in\mathbb C^{n\times n},
\]
the {\it low-rank Toeplitz matrix approximation} problem is to find a rank-$r$ Toeplitz matrix $\tilde T$ that approximates $T$ (with respect to either the spectral or Frobenius norm). Of course, this problem is only nontrivial due to the requirement that $\tilde T$ is also Toeplitz. The same question can be asked for a Hankel matrix $H$ instead of Toeplitz $T$, and they are entirely equivalent problems\footnote{Indeed, if $J$ denotes the exchange matrix,
%\[
%        ( J\bfz)_i=z_{n-1-i},\qquad i=0,\dots,n-1,
%\]
then $T J$ is Hankel whenever $T$ is Toeplitz.}.

Both Toeplitz and Hankel matrices encode
translation-invariant sampling relations, and are fundamental in deconvolution, system
identification, array processing, inverse scattering, spectral compressed sensing, and related
imaging problems, see \cite{hansen2002deconvolution,paduart2010identification,chen2013spectral,wang2016seismic,ongie2017fast,ying2018vandermonde,li2021stable,xu2021sparse,lawal2021toeplitz}. One underlying principle that is exploited in these works is that Toeplitz and Hankel matrices are often approximately low-rank and have additional algebraic structure. Although this computational approximation has received substantial interest from a wide audience, see \cite{cybenko1982moment,rosen1996total,chu2003structured,gillard2023hankel,musco2024sublinear}, we are not aware of an optimal algorithm for this problem. 

\item 
Given measurement $\tilde\bfy$, the {\it approximation of exponential sums from noisy data} problem is to find an exponential sum $\tilde f$ with exactly $r$ terms that approximates $f$ (in either the $L^2([0,n])$ or $\ell^2(\{0,1,\dots,n-1\})$ norm). This is a nontrivial problem due to the first requirement that $\tilde f$ is an exponential sum, which is not a linear space of functions. 

Interest in this problem can be motivated by the classical Whittaker–Nyquist–Shannon sampling theorem which concerns bandlimited functions sampled on the integers. In many modern inverse problems, however, the unknown object is not just bandlimited, but also described by a small number of spectral parameters. 
\end{enumerate}

\begin{figure}[ht]
	\centering
	\begin{tikzpicture}[scale=1.0, every node/.style={align=center}]
		\node (U) at (0,3.2)
		{\begin{tabular}{c}
				Fourier subspace\\
				estimation
		\end{tabular}};
		
		\node (M) at (-3.8,0)
		{\begin{tabular}{c}
				Structured low-rank\\
				%Toeplitz/Hankel\\
				approximation
		\end{tabular}};
		
		\node (S) at (3.8,0)
		{\begin{tabular}{c}
				Sampling / interpolation\\
				by exponential sums
		\end{tabular}};
		
		\node[draw, rounded corners, inner sep=2pt] (E) at (0,1.25)
		{\begin{tabular}{c}
				Spectral estimation
		\end{tabular}};
		
		\draw[<->, thick] (S) -- (M)
		node[midway, below]
		{\scriptsize Hankel/Toeplitz lift};
		
		\draw[<->, thick] (S) -- (U)
		node[midway, right, xshift=3pt]
		{\scriptsize signal subspace};
		
		\draw[<->, thick] (M) -- (U)
		node[midway, left, xshift=-3pt]
		{\scriptsize leading eigenspace};
		
		\draw[->, thick, dashed] (E) -- (U);
		\draw[->, thick, dashed] (E) -- (M);
		\draw[->, thick, dashed] (E) -- (S);
	\end{tikzpicture}
	\caption{Spectral estimation as the core engine for three structured approximation problems}
	\label{fig:relationship}
\end{figure}

\begin{figure}[ht]
    \centering
    \begin{tikzpicture}[scale=1.0, every node/.style={align=center}]
        \node[draw, rounded corners, inner sep=7pt] (N) at (-4,0)
        {\begin{tabular}{c}
            Empirical subspace $\widetilde U$\\
            %\(\widetilde U\in \Gr(r,\mathbb C^n)\)\\[1mm]
            \(\|\sin(\widetilde U,U)\|_2=\eta\)\\
            \(\|\sin(\widetilde U,U)\|_{\rm F}=\eta\)
        \end{tabular}};

        \node[draw, rounded corners, inner sep=7pt] (F) at (4,0)
        {\begin{tabular}{c}
            Spectral error \\
            %\(U(\widehat{\bfx})\in \mathcal F_{n,r,h}\)\\[1mm]
            \(\|\widetilde{\bfx}-\bfx\|_\infty\lesssim \eta/n\)\\
             \(\|\widetilde{\bfx}-\bfx\|_2\lesssim \eta/n\)
        \end{tabular}};

        \draw[->, thick] (N) -- (F)
            node[midway, above]
            { Gradient-MUSIC};

        %\node at (0,-0.8)
        %{\scriptsize subspace error \(\eta\) is transferred into frequency error \(\eta/n\)};
    \end{tikzpicture}
    \caption{Optimal spectral estimation by Gradient-MUSIC.}
    \label{fig:gmusic-subspace-to-fourier}
\end{figure}

A relationship among the three problems and spectral estimation may be summarized by \cref{fig:relationship}. Here the outer triangle represents the other three structured approximation problems addressed in this paper.  Sampling/interpolation is the source: finite samples of exponential sums encode the continuous parameters.  Structured
low-rank matrix approximation arises by arranging the samples into Toeplitz or Hankel matrices. Fourier subspace estimation arises by extracting the Vandermonde column space.  Spectral
estimation sits at the center because, in the well-separated regime, it converts approximate subspace information into accurate frequency information, which can then be used to solve the other two structured approximation problems.

	\subsection{Notation}
	Generally vectors and matrices are written in bold. We use subscripts to denote an entry of a vector or matrix, e.g., $a_j$ and $A_{j,k}$ denote the $j$-th entry of a vector $\bfa$ and $(j,k)$-th entry of a matrix $ A$. 
	
	For $\bfu\in \C^n$, let $\diag(\bfu)$ be the $n\times n$ diagonal matrix with $\bfu$ as its diagonal entries, presented in the same order as $\bfu$. For subspaces $ U$ and $V$ in $\C^n$ of the same dimension, let $\sin( U,V)$ be the diagonal matrix that contains the sines of their principal angles.  
	
	For a matrix $ A$, let $\| A\|_2$ and $\| A\|_{\rm F}$ denote its spectral and Frobenius norm, respectively. The Moore-Penrose inverse of $A$ is denoted $A^\dagger$ and conjugate transpose of $A$ is denoted $A^*$. We use the shorthand notation $A^{\dagger,*}:=(A^\dagger)^*$. For a vector $\bfu$, let $\|\bfu\|_p$ denote its $p$-norm, for $p\in [1,\infty]$. We write $a\lesssim b$ if there is a universal constant $C>0$ such that $a\leq Cb$. 
	
	We say $z$ is a complex normal random variable with mean zero and variance $\sigma^2$ if $z=x+iy$ where $x,y$ are independent normal random variables that have mean zero and variance $\sigma^2/2$.

\section{Main results and contributions}

The main contribution of this paper is a quantitative transference principle based on spectral estimation by Gradient-MUSIC,  as depicted in Figure \ref{fig:gmusic-subspace-to-fourier},   which is quantitatively related to the other three problems. 
\begin{enumerate}
	\item 
	(Upper bound transference mechanism) We show how the Gradient-MUSIC algorithm and theory can be readily modified to produce methods and performance guarantees for the other three structured approximation problems. Schematically,
	\[
\text{Empirical noise}
\longrightarrow
{\rm subspace \,\,error}
\longrightarrow
{\rm spectral \,\, error}
\longrightarrow
\text{reconstruction error}.
\]
	\item 
	(Lower bound transference mechanism) Perhaps surprisingly, we prove there is a mechanism of converting a minimax lower bound for spectral estimation into an analogous one for the other three problems. This is proved by showing that if any of the three problems could be solved more accurately than possible, then composing such a hypothetical method with Gradient-MUSIC would produce a spectral estimation algorithm whose performance contradicts known minimax lower bounds. Schematically, 
	\[
	\text{too-good structured estimator}
\longrightarrow
\text{too-good Fourier subspace}
\longrightarrow
\text{too-good spectral estimator}.
\]
	
	\item 
	(Optimality) The methods for all three structured approximation problems are provably minimax optimal in the dimension/samples $n$, rank/sparsity $r$, and noise level $\|\bfz\|_2$ or $\|E\|_2$. 
\end{enumerate}

Thus spectral estimation is not only an algorithmic tool but also the quantitative obstruction governing these structured approximation problems. One strength of our approach is a balance between theory and computation. Our methods are efficient, relatively simple to implement, come with strong performance guarantees, and are provably optimal in the well-separated regime. 

For Fourier subspace estimation, a natural metric is $\|\sin(U,V)\|_{\rm F}$, the Frobenius distance between subspaces $U,V$. \cref{thm:Fouriersubspace} shows that, under the well-separated assumption, given a single observation of a Fourier subspace $U\subset\C^n$ of dimension $r$ which has been perturbed by noise $\bfz\in\C^n$ such that $\|\bfz\|_2\leq\alpha$ for small enough absolute $\alpha>0$, our method outputs a Fourier subspace $U_\sharp\subset\C^n$ of dimension $r$ such that for some absolute $C>0$,
	$$
	\norm{\sin( U,U_\sharp)}_{\rm F}
	\leq \frac C {\sqrt n} \, \|\bfz\|_2.
	$$
	We also prove that this approximation rate is minimax optimal in $n$, $r$, and $\|\bfz\|_2$, which means that no other method can scale better with respect to these parameters under the same assumptions. We also prove a bi-Lipschitz bound for distances between $\bfx,\tilde\bfx$ and $U(\bfx),U(\tilde\bfx)$. 
	%To our best knowledge, this is the first result with such performance guarantees.

	For the Toeplitz matrix problem, suppose $T\in \C^{n\times n}$ is a rank $r$ Toeplitz matrix that satisfies a regularity assumption which controls $\sigma_1(T)/\sigma_r(T)$, where $\sigma_1$ and $\sigma_r$ are, respectively, the largest and smallest positive singular value. \cref{thm:gdmusictoeplitz} shows that for $\| E\|_2\leq \alpha n$ for sufficiently small absolute $\alpha>0$, our algorithm outputs a Toeplitz matrix $T_\sharp$ of rank exactly $r$ such that for some absolute $C>0$,
	$$
	\norm{ T-T_\sharp}_2 \leq C \| E\|_2.
	$$
	We also prove that this approximation rate is minimax optimal in $n$, $r$, and $\| E\|_2$, which means that no other method can scale better with respect to these parameters under the same assumptions. Perhaps interestingly, the upper bound does not depend on $n$ except implicitly through $\| E\|_2$. 
	
	For the final problem, suppose we are given noisy samples of $f$ on $\{0,1,\dots,n-1\}$ and $f$ has exactly $r$ exponentials whose frequencies are well-separated. \cref{thm:gradmusicexponentialsum} states that if $\bfz\in\C^n$ represents the noise and $\|\bfz\|_2\leq \alpha \sqrt n$ for a sufficiently small absolute $\alpha>0$, then our method produces an exponential sum $f_\sharp$ with exactly $r$ exponentials such that for some absolute $C>0$, 
	\begin{align*}
		\Bigg( \sum_{j=0}^{n-1} \big| f(j) -  f_\sharp(j) \big|^2 \Bigg)^{1/2}
		\leq C \|\bfz\|_2.
	\end{align*}
	This performance guarantee is minimax optimal  in $n$, $r$ and $\|\bfz\|_2$, which means that no other method can scale better with respect to these parameters under the same assumptions.

\subsection{Comparison with prior work}
\label{sec:relatedworks}

It is known that spectral estimation and the other three structured approximation problems are connected. For instance, \cite{cybenko1982moment} used the Toeplitz problem as motivation for Fourier subspace estimation, classical subspace-based methods like MUSIC start by finding a reasonable estimate for a Fourier subspace by a general subspace \cite{krim1996two}, denoising iterative methods such as \cite{cadzow1988signal} first approximates the data by a low-rank Hankel matrix to improve spectral estimation algorithms, and \cite{beylkin2005approximation,plonka2019computation} use spectral estimation algorithms to approximate exponential sums and related functions. 

The results of this paper go beyond prior work and illustrate that there are fundamental {\it quantitative} connections between these three problems. The present paper does not attempt to survey all of these directions. Instead, it isolates a particular regime (well-separated frequencies) in which spectral estimation gives a quantitative transference principle among the three vertices of the triangle. 

Fourier subspace estimation appears to have no direct predecessor. In \cref{sec:discussion1}, we discuss a connection between Gradient-MUSIC and a generic optimization algorithm on the Grassmannian.

Approximation by low-rank Toeplitz and Hankel matrices has a substantial literature, including
alternating projection methods, structured total least norm methods, and approaches based on
sparse inverse Fourier transforms.  These methods are important and often effective, but they do
not always return a matrix of the correct rank, and in several cases sharp perturbation guarantees
are unavailable. A detailed and technical comparison can be found in \cref{sec:discussion2}.

Our results for the approximation of exponential sums from noisy data can be thought of as an analogue of the Whittaker–Nyquist–Shannon sampling theorem. The main difference is that the class of exponential sums with $r$ frequencies is not a linear subspace unlike bandlimited functions and consequently, our reconstruction method is (necessarily) nonlinear unlike the classical sampling theorem. More details can be found in \cref{sec:discussion3}. 

The reader will observe that one could have used another spectral estimation algorithm (e.g., ESPRIT, Prony's, matrix pencil, convex optimization, greedy methods) instead of Gradient-MUSIC for the three other structured approximation problems. In this paper, we used Gradient-MUSIC for two reasons. First is theoretical, as it is provably optimal in the number of samples $n$, sparsity $r$, and noise level $\|\bfz\|_p$ (for deterministic noise $\bfz$ and any $p\in[1,\infty]$) in the well-separated case. Using any other spectral estimation algorithm may result in a worse rate for the other three problems. Additionally, Gradient-MUSIC has an additional property that its performance is tied purely to the subspace error, see the discussion after \cref{thm:maingradientMUSIC}, a property that is required in several proofs. Second is computational, as extensive numerical simulations in \cite{liao2016music} showed that classical MUSIC has excellent performance in the well-separated setting compared to other algorithms, and Gradient-MUSIC is a strictly more efficient algorithm with the same performance guarantees (which are minimax optimal shown in \cite{fannjiang2025optimality}). 

There is also a rich statistical theory of spectral estimation. Both ESPRIT \cite{ding2024esprit} and Gradient-MUSIC are optimal for i.i.d. sub-Gaussian noise. The latter also has performance guarantees for nonstationary Gaussian noise which (we believe) are also optimal. By using the techniques developed in this paper, there is an analogous statistical theory for these three structured approximation problems. We have focused only on the deterministic theory to keep the scope of this paper reasonable.

\subsection{Organization}

\cref{sec:review} contains background information on spectral estimation problem, Fourier matrices and their subspaces, and performance guarantees for Gradient-MUSIC. 

Importantly, the remaining three structured approximation problems are contained in separate sections, which can be read independently of each other. 
\begin{itemize} \itemsep-2pt
	\item 
	\cref{sec:Fouriersubspace} deals with Fourier subspace estimation.
	\item 
	\cref{sec:toeplitz} studies approximation by low-rank Toeplitz matrices.
	\item 
	\cref{sec:exponentialsums} is on approximation of exponential sums from noisy data. 
\end{itemize}
Each of these sections has the same organization: a subsection on problem formulation, main approximation theorems accomplished by Gradient-MUSIC, minimax lower bounds, and a discussion section which also includes comparison to prior work. 

Since the Toeplitz matrix problem is the most well-studied of the three, numerical simulations which compare to alternating projection are in \cref{sec:numerics}. Finally, conclusions can be found in \cref{sec:conclusion}.

\section{Background}
\label{sec:review}

\subsection{Fourier matrices and subspaces}

Through the paper, $n$ will be a fixed integer. Even though various definitions will implicitly depend on $n$, we will suppress their dependence on $n$ to reduce clutter.

The minimum separation of $\bfx$ controls the stability to noise of spectral estimation. Indeed, \cite[Theorem 1]{aubel2019vandermonde} showed that for any $\beta>1$,
\begin{equation}\label{eq:Fouriersigs}
	\Delta(\bfx)\geq \frac {2\pi \beta}n 
	\quad \Longrightarrow \quad 
	\sqrt{n(1-\beta^{-1})} \leq \sigma_{\min}(\Phi_n(\bfx))\leq \sigma_{\max}(\Phi_n(\bfx))\leq \sqrt{n(1+\beta^{-1})}.
\end{equation}
Note that the separation condition $\Delta(\bfx)\geq 2\pi \beta/n$ implicitly implies that $n\geq \beta r$ where $r$ denotes the cardinality of $\bfx$.  By making a slightly strong assumption that $\beta\geq 2$, then 
$$
\Delta(\bfx)\geq \frac {4\pi}n  \quad \Longrightarrow \quad 
\sigma_{\min}(\Phi_n(\bfx))\asymp \sigma_{\max}(\Phi_n(\bfx))\asymp \sqrt n.
$$  

It will be helpful to define the following parameter sets,
\begin{align}
	\Omega(h,r)
	&:=
	\Bigl\{
	\bfx\in \mathbb T^r:
	\Delta(\bfx):=\min_{j\neq k}|x_j-x_k|_{\mathbb T}\ge h
	\Bigr\} \label{eq:separated}, \\
	\calP(h,r)
	&:=\Big\{(\bfx,\bfa)\colon \bfx\in \Omega(h,r), \, 1\leq |a_j|\leq 10 \text{ for $j=1,\dots,r$}\Big\}. \label{eq:parameterset}
\end{align}
Note that if $h=2\pi \beta/n$ for an integer $\beta$, then these sets are nonempty if and only if $n\geq \beta r$. In the second definition, the `1' and `10' have no special meaning, as they simply ensure that the dynamic range of the amplitudes cannot be too large. They can be replaced with arbitrary positive $r_0\leq r_1$, though some of the absolute constants that appear throughout the paper will depend on $r_0$ and/or $r_1$.

The Fourier subspace associated with $\Phi_n(\bfx)$ is defined as 
$$
U(\bfx):=\range(\Phi_n(\bfx))\subset\C^n. 
$$
Whenever $n\geq 2r$, there is a bijection between all sets $\bfx\subset\T$ of cardinality $r$ with Fourier subspaces of dimension $r$ in $\C^n$, as shown in \cref{lem:bijection}. 

When a  subspace $W$ is canonically identified with its orthogonal projector $P_W$, the Grassmannian $\Gr(r,\mathbb C^n)$ of $r$-dimensional subspaces of $\mathbb C^n$, 
can be identified with the space of orthogonal projectors 
\[
\Gr(r,\mathbb C^n)\cong\big\{P\in \C^{n\times n}: P^2=P,\,\, \hbox{$P$ is Hermitian}\big\}.
\]
In the case of a Fourier subspace $U$, $P_U$ can be canonically represented by the  Hermitian matrix
\[
P(\bfx):=\Phi_n(\bfx)\bigl(\Phi_n(\bfx)^*\Phi_n(\bfx)\bigr)^{-1}\Phi_n(\bfx)^*.
\]
The Grassmannian $\Gr(r,\C^n)$ is embedded into the space of Hermitian matrices
endowed  with  the real Hilbert--Schmidt inner product on
Hermitian matrices
\[
\langle A,B\rangle_{\mathbb R}
:=
\operatorname{Re}\operatorname{tr}(A^*B)=\operatorname{tr}(AB).
\]
The resulting distance, called the {\it Frobenius distance},  
between $U,V\in \Gr(r,\C^n)$ is given by
\begin{equation}
	\label{eq:Chordal}
	\norm{\sin( U, V)}_{\rm F}
	=\frac{1}{\sqrt{2}}\|P_U-P_V\|_{\rm F}
	=\Big(\sum_{j=1}^r\sin^2 (\theta_j)\Big)^{1/2},
\end{equation}
where  $\{\theta_j\}_{j=1}^r$  are  the principal angles between the subspaces $U$ and $V$. An alternative, globally defined metric is the {\it spectral distance}, 
\begin{equation}
	\label{eq:sinetheta}	
	\norm{\sin( U, V)}_2=\|P_U-P_V\|_2
	=\max_{j=1,\dots, r}\sin (\theta_j).
\end{equation}
These relationships can be found in the monograph \cite{stewart1990matrix}. 

Given any vector $\bfu\in \C^n$, set $m:=\lceil n/2\rceil$ and define the $m \times m$ Hankel matrix of $\bfu$ as
\begin{equation} \label{eq:hankel}
	H(\bfu):=
	\begin{bmatrix}
		u_0 &u_1 &\cdots &u_{m-1} \\
		u_1 &u_2 &\cdots &u_{m} \\
		\vdots &\vdots & &\vdots \\
		u_{m-1} &u_{m} &\cdots &u_{2m-2}
	\end{bmatrix}.
\end{equation}
Given $\bfy=\Phi_n(\bfx)\bfa\in\C^n$, a direct computation establishes the Fourier matrix factorization 
\begin{equation}\label{eq:hankelfactorization}
	H(\bfy)=\Phi_m(\bfx)\diag(\bfa)\Phi_m(\bfx)^T.
\end{equation}

	\begin{remark}[Thresholding] \label{rem:threshold}
		The rank of $H(\bfy)$ defined in \eqref{eq:hankelfactorization} can be deduced from noisy data due to a general singular value thresholding in result \cite[5.6]{fannjiang2025optimality}. Suppose $\Delta(\bfx)\geq 2\pi \beta/n$ for $\beta > 1$ and $r_0\leq |a_j|\leq r_1$ for all $j$. Given $\tilde\bfy=\bfy+\bfz$ for any $\bfz$ such that $\|\bfz\|_2\leq \alpha \sqrt n$ for a small enough $\alpha>0$, consider the perturbed Hankel matrix $H(\tilde \bfy)$. The rank $r$ is precisely the largest $k$ such that 
		$$
		\frac{\sigma_k(H(\tilde\bfy))}{\sigma_1(H(\tilde\bfy))}
		\geq \frac 7 {10} \frac{r_0}{r_1} \frac{\beta-1}{\beta+1}. 
		$$	
	\end{remark}	

	\subsection{Review of Gradient-MUSIC}

	The goal of spectral estimation is to estimate the unknown frequencies $\bfx$ of $f$ given samples $\tilde\bfy$. To quantify the accuracy, if $\tilde \bfx\subset \T$ is any other set consisting of $r$ elements as well, then the $\ell^p$ error for $p\in [1,\infty]$ is defined as
	$$
	\|\tilde \bfx - \bfx \|_p
	:= \min_{\text{permutation }\sigma } \ \ \left( \sum_{k=1}^r | x_k - \tilde x_{\sigma(k)}|_\T^p \right)^{1/p}.
	$$
	This is the usual matching distance between $\bfx$ and $\tilde \bfx$ when the $\ell^p$ metric on $\T$ is used. Throughout, we implicitly assume that $\bfx$ and $\tilde \bfx$ are indexed so that the identity permutation minimizes the matching distance. 
	
	We developed and analyzed Gradient-MUSIC in our prior work \cite{fannjiang2025optimality}. We created several versions of this method that can be used for different tasks. In this paper, Gradient-MUSIC refers to \cite[Algorithm 2]{fannjiang2025optimality}, the one that takes a subspace $ W$ (not necessarily Fourier) which encodes some information about the function samples and outputs an approximation of $\bfx$. It can be viewed as a nonconvex optimization algorithm with performance guarantees. Let us briefly explain how the Gradient-MUSIC algorithm works. Define the ``steering vector" $\bfphi\colon \T\to \C^n$ by 
		$$
		\bfphi(t)
		:= \frac 1 {\sqrt n} \Big[e^{ijt}\Big]_{j=0,\dots,n-1}. 
		$$
		Note that $\bfphi$ has been normalized to have unit length. Consider the two {\it MUSIC functions},
		\begin{equation}
			\label{eq:musicfunctions}
			q(t)= 1 - \bfphi(t)^* UU^* \bfphi(t), \andspace
			\tilde q(t)= 1 - \bfphi(t)^* WW^* \bfphi(t). 
		\end{equation}
		The first function is the noiseless version corresponding to the true Fourier subspace $U$, while the second is defined with the approximate subspace $W$ which acts as a perturbation of $U$. Suppose $\norm{\sin( U, W)}_2\leq 0.01$. By \cite[Theorem 5.9]{fannjiang2025optimality}, it holds that: 
		\begin{enumerate}[(a)]
			\item 
			The MUSIC function $\tilde q$ has at least $r$ local minima, and if the smallest $r$ local minima are $\tilde\bfx=\{\tilde x_k\}_{k=1}^r$, then they satisfy the error bound
			\begin{equation}
				\label{eq:minimaperturbation}
				\|\tilde\bfx-\bfx\|_\infty \leq \frac{7} n \norm{\sin( U, W)}_2. 
			\end{equation}
			\item 
			For $k=1,\dots,r$, the local minimum $\tilde x_k$ is the only critical point of $\tilde q$ in the interval
			$$
			\left[\tilde x_k-\frac{4\pi}{3n}, \, \tilde x_k+ \frac{4\pi}{3n}\right]
			$$
			and gradient descent initialized in this interval with step size $c/n^2$ for an explicit absolute $c>0$ converges at a linear rate to $\tilde x_k$. 
			\item 
			Initialization can be found by thresholding $\tilde q$ on any subset of $\T$ with mesh norm at most $1/(2n)$. For example, a uniform grid on $\T$ with width at most $1/(2n)$ suffices.
		\end{enumerate}
		
		\begin{figure}
			\centering
			\includegraphics[width=0.9\textwidth]{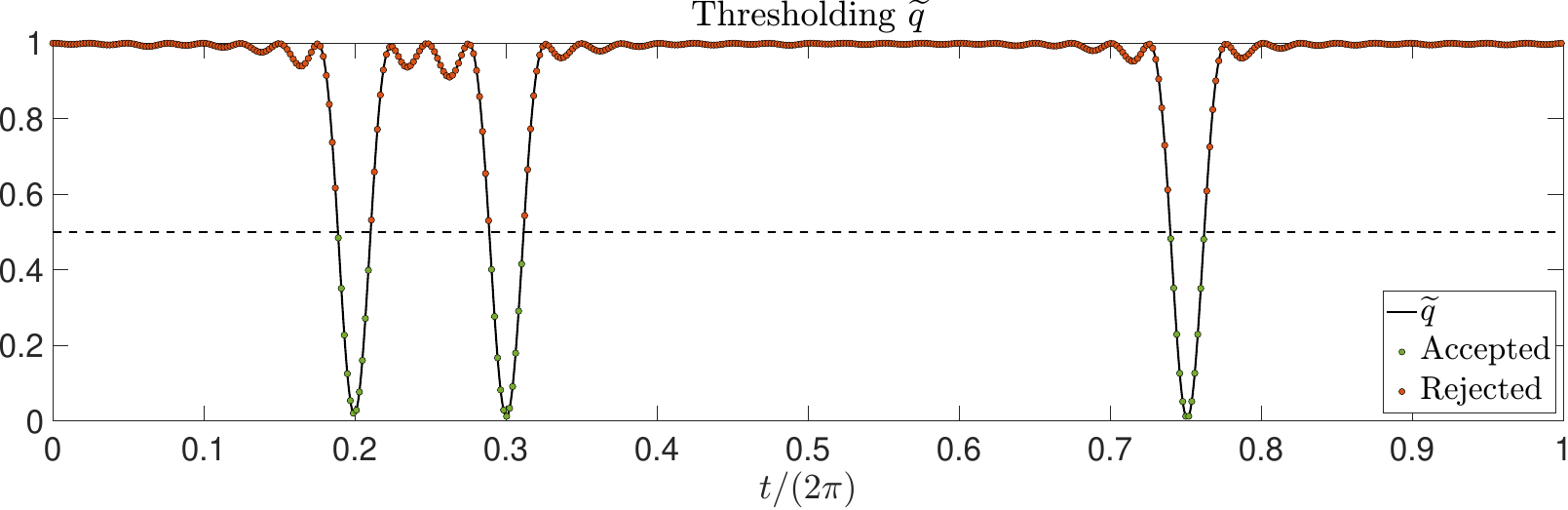}
			\caption{An example of a MUSIC function $\tilde q$ where $n=40$ and $\bfx=\{2\pi(0.2), 2\pi(0.3), 2\pi(0.75)\}$. Thresholding $\tilde q$ on a grid of width $1/(2n)$ and keeping grid points $t$ where $|q(t)|\leq 0.529$ provably finds points in the basin of attraction for each local minimum $\tilde x_k$ of $\tilde q$.}
			\label{fig:landscapeexample}
		\end{figure}
		An example of a $\tilde q$ is shown in \cref{fig:landscapeexample}. The two main steps in Gradient-MUSIC are finding initialization by thresholding and running gradient descent. 
		
		Gradient-MUSIC is a modification of the classical MUSIC algorithm \cite{schmidt1986multiple}. Under the assumptions of \cref{thm:maingradientMUSIC}, they have the same approximation guarantees (both aim to find the smallest $r$ local minima $\tilde\bfx$ of $\tilde q$), but the former is significantly faster. This is made possible by the special geometric properties of the MUSIC function which allow Gradient-MUSIC to efficiently find initialization (regardless of the noise level) for fast local optimization methods such as gradient descent. In contrast, classical MUSIC requires a brute force search over a fine grid, which makes it inefficient for settings where it is possible to estimate the frequencies with great accuracy.

	The following is a performance guarantee summarizing the properties of Gradient-MUSIC.
	
	\begin{theorem}[Theorems 5.9 and 5.11, \cite{fannjiang2025optimality}] \label{thm:maingradientMUSIC}
		Let $n\geq 100$. Let $\bfx\subset \T$ such that $\Delta(\bfx)\geq 8\pi/n$ and $U=\range(\Phi_n(\bfx))\subset\C^n$ be its associated Fourier subspace. For any subspace $ W$ of the same dimension as $ U$ such that $\|\sin( U, W)\|_2\leq 1/100$, the MUSIC function $\tilde q:=1-\bfphi(t)^*WW^*\bfphi(t)$ has at least $r$ local minima such that if $\tilde\bfx$ denotes the $r$ smallest ones, then  
			\begin{equation*}
				\|\tilde\bfx-\bfx\|_\infty
				\leq \frac {7} n \norm{\sin( U, W)}_2.
			\end{equation*}
			Furthermore, Gradient-MUSIC takes $ W$ and outputs $\bfx^\sharp$ of the same cardinality as $\bfx$ such that the frequency error satisfies
		\begin{align*}
			\|\bfx^\sharp-\bfx\|_\infty
			&\leq \frac {10} n \, \norm{\sin( U, W)}_2.
		\end{align*}
	\end{theorem}
	
	One strength of this theorem is its generality. We emphasize that $W$ is any arbitrary subspace, which is not necessarily of Fourier type. The only assumption is $\|\sin(U, W)\|_2\leq 1/100$, so the theorem does not take into account how $ W$ is computed. 
	
	For the first step, it is sufficient to evaluate $\tilde q$ on any (potentially nonuniform) grid with mesh norm at most $1/(2n)$, which has worst case complexity $C n^2r$. Letting $\epsilon$ be any quantity such that $\norm{\sin( U, W)}_2\leq \epsilon$, one needs at most $C\log(1/\epsilon)$ gradient iterations. Then the upper bound in \cref{thm:maingradientMUSIC} becomes 
	$$
	\|\bfx-\bfx^\sharp\|_\infty
	\leq \frac{10 \epsilon} n,
	$$
	and the worst case computational complexity of Gradient-MUSIC is at most
	\begin{equation}\label{eq:computationalcost}
		C\left(n^2 r + n r^2 \log\left( 1/\epsilon \right)\right).
	\end{equation}
	This algorithm is compatible with parallelization. With $p$ parallel cores such that $r\leq p\leq n$, the worst case time complexity is at most 
	\begin{equation} \label{eq:computationalcostparallel}
		C\left( n^2 r p^{-1} + n r\log\left( 1/\epsilon \right) \right).
	\end{equation}
	
	\begin{remark}[Reduced complexity using FFT]
		If the grid is chosen to be uniform, then evaluation of $\tilde q$ on the grid can be reduced to $O(n\log(n)r)$. However, the methods that we will use for all three structured approximation algorithms will also require a truncated SVD of a Hankel matrix which has complexity $O(n^2 r)$, though this performance can be improved with randomized SVD method, with the trade-off that there is a chance of failure.
	\end{remark}
	
	\subsection{A modified result for Gradient-MUSIC}
	
	\cref{thm:maingradientMUSIC} controls the $\ell^\infty$ error between $\bfx^\sharp$ and $\bfx$ via the spectral distance of $U$ and $W$. Since the latter is the largest principal angle between $U$ and $W$, this is a $\ell^\infty$ to $\ell^\infty$ estimate. 
	
	For several results, it will be necessary control the error between $\bfx^\sharp$ and $\bfx$ in $\ell^2$. Directly applying \cref{thm:maingradientMUSIC} comes at a price of $\sqrt r$, which we would like to avoid. Instead, we will control the $\ell^2$ error in terms of the Frobenius distance between $U$ and $W$, which is a $\ell^2$ sum of their principal angles. Consequently, the following theorem is a $\ell^2$ to $\ell^2$ bound and is proved in \cref{proof:maingradientMUSIC2}. 
	
	\begin{theorem} \label{thm:maingradientMUSIC2}
		Let $n\geq 100$. Let $\bfx\subset \T$ such that $\Delta(\bfx)\geq 8\pi/n$ and $ U=\range(\Phi_n(\bfx))\subset\C^n$ be its associated Fourier subspace. For any subspace $ W$ of the same dimension as $ U$ such that $\|\sin( U, W)\|_{\rm F}\leq 1/100$, the MUSIC function $\tilde q:=1-\bfphi(t)^*WW^*\bfphi(t)$ has at least $r$ local minima such that if $\tilde\bfx$ denotes the $r$ smallest ones, then  
		\begin{equation}
			\label{eq:goal0}
			\|\tilde\bfx-\bfx\|_2 
			\leq \frac {20} n \norm{\sin( U, W)}_{\rm F}.
		\end{equation}
		Furthermore, Gradient-MUSIC takes $ W$ and outputs $\bfx^\sharp$ of the same cardinality as $\bfx$ such that the frequency error satisfies
		\begin{align*}
			\|\bfx^\sharp-\bfx\|_2
			&\leq \frac{25} n \, \norm{\sin( U, W)}_{\rm F}.
		\end{align*}
	\end{theorem}
	
	\cref{thm:maingradientMUSIC} and \cref{thm:maingradientMUSIC2} obviously do not imply each other. By interpolating between the two theorems, we get $\ell^p$ bounds for the error between $\bfx$ and $\bfx^\sharp$. 
	
	The computational complexity of Gradient-MUSIC for finding $\bfx^\sharp$ satisfying \cref{thm:maingradientMUSIC2} is  \eqref{eq:computationalcost} and \eqref{eq:computationalcostparallel}, except $\epsilon$ is any quantity such that $\norm{\sin( U, W)}_{\rm F}\leq \epsilon$.

	\subsection{A minimax lower bound for spectral estimation}
	
	Fix $\epsilon>0$ and $p\in [1,\infty]$. Let $\psi$ be any function that maps any given data (any possible $\tilde \bfy$ such that $\tilde \bfy =\Phi_n(\bfx)\bfa + \bfz\in \C^n$ for $(\bfx,\bfa)\in \calP(16\pi/n,r)$ and $\|\bfz\|_p\leq \epsilon$) to a set $\bfx^\psi\subset\T$ of cardinality $r$. Here, $\psi$ does not have to be computationally efficient or even computable. The minimax error is 
	\begin{equation}
		\label{eq:minimaxspectral}
		X_*(n,r,p,\epsilon)
		=\inf_\psi \sup_{(\bfx,\bfa)\in \calP(16\pi/n,r)} \,  \sup_{\|\bfz\|_p \leq \epsilon } \, \|\bfx-\bfx^\psi\|_\infty,
	\end{equation}
	where the inf is taken over all possible such functions $\psi$. This is a lower bound for all possible methods, including Gradient-MUSIC.
	\begin{lemma}[\cite{fannjiang2025optimality} Lemma 4.3] \label{lem:minimax}
		For any $n,r\geq 1$, $p\in [1,\infty]$, and  $\epsilon\leq 8\pi n^{1/p}$, we have
		$$
		X_*(n,r,p,\epsilon)
		\geq \frac{\epsilon}{4n^{1+1/p}}.
		$$
	\end{lemma}

	\section{Fourier subspace estimation}
	\label{sec:Fouriersubspace}
	
	\subsection{Problem formulation} 
	
	We begin by setting up the Fourier subspace estimation problem in a careful manner. Suppose we observe a vector of the form $\bfy+\bfz \in \C^{n}$ where the noise $\bfz$ is small enough and for some $(\bfx,\bfa)$, 
	\begin{equation}
		\label{eq:Fouriersubspace}
		\bfy = \Phi_n(\bfx)\bfa.
	\end{equation}
	Thus, the $\bfy$ lies in a Fourier subspace $U(\bfx):=\range(\Phi_n(\bfx))$. Can we find a Fourier subspace $\tilde U\subset\C^n$ of the same dimension that approximates the underlying $U$ with high precision? By approximation, we use either 
	$$
	\norm{\sin(U,\tilde U)}_2 \quad \text{or} \quad \norm{\sin(U,\tilde U)}_{\rm F}.
	$$
	
	At first glance this problem may seem impossible. Even if there is no noise and we get access to $\bfy=\Phi_n(\bfx)\bfa$, general linear algebra tells us that we cannot deduce a $r$ dimensional subspace unless we observe at least $r$ linearly independent vectors in it. However, we argue that assuming $n\geq 2r$ and all entries of $\bfa$ are nonzero imply that the noiseless problem is solvable. 
	
	\begin{lemma}\label{lem:bijection}
		If $n\geq 2r$, $\bfx\subset\T$ has cardinality $r$, all entries of $\bfa\in \C^r$ are nonzero, and $\bfy=\Phi_n(\bfx)\bfa$, then there is no other $\bfx'$ with cardinality $r$ such that $\bfy=\Phi_n(\bfx')\bfa'$ for some $\bfa'$. 
	\end{lemma}
	
	\begin{proof}
		Indeed, suppose for the purpose of deriving a contradiction that $\bfy=\Phi_n(\bfx')\bfa'$ for some $(\bfx',\bfa')$ such that $\bfx'\not=\bfx$ and $\bfx'$ has cardinality $r$. Then we end up with the system of equations 
		$$
		[\Phi_n(\bfx),\Phi_n(\bfx')][\bfa,-\bfa']^T=\bfzero.
		$$
		Deleting any columns of $[\Phi_n(\bfx),\Phi_n(\bfx')]$ that are repeated and redefining $[\bfa,-\bfa']$ accordingly, we obtain the system
		$$
		\Phi_n(\bfx\cup\bfx')\bfb = \bfzero.
		$$
		Since all entries of $\bfa$ are nonzero and $\bfx\not=\bfx'$, we see that $\bfb$ cannot be identically zero. This shows that the Fourier matrix $\Phi_n(\bfx\cup \bfx')$ with $n$ rows and at most $2r$ columns has a nontrivial null space. This contradicts the Vandermonde determinant theorem since $n\geq 2r$.
	\end{proof}
	
	Throughout, we will implicitly assume that $n\geq 2r$ and that all entries of $\bfa$ are nonzero. Even though this lemma suggests that Fourier subspace estimation may be possible for small noise, it is not clear how to do this efficiently. The main obstacle is that $\tilde U$ is also required to be a Fourier subspace. Fourier subspace estimation is in disguise, a computational differential geometry problem. Each Fourier subspace is a point in the Grassmannian manifold consisting of all subspaces of dimension $r$ in $\C^n$. The set of all Fourier subspaces is a sub-manifold of dimension $r$. Given a $\bfy+\bfz\in\C^n$ that is close enough to a Fourier subspace, the problem is to ``project" it to the sub-manifold of Fourier subspaces. 
	
	By a scaling argument, without loss of generality, we assume that $|a_j|\geq 1$ for all $j$. In order to study Fourier subspace estimation properly, we will restrict our attention to
	$$
	(\bfx,\bfa)\in \calP\left(16\pi/n, r\right).
	$$
	While it would be interesting to relax the minimum separation requirement in this definition, $\Phi_n(\bfx)$ may be extremely ill-conditioned due to various results about Fourier matrices, see \cite{li2021stable,batenkov2020conditioning,kunis2020smallest,batenkov2021spectral,barnett2022exponentially,li2025multiscale}. This would greatly change the nature of this problem. 
	
	Finally, under these conditions, inequality \eqref{eq:Fouriersigs} implies that $\sigma_r(\Phi_n(\bfx))\gtrsim \sqrt n$ and so $\|\bfy\|_2\gtrsim \sqrt n$. Consequently, it is impossible to estimate $ U$ unless 
	$$
	\|\bfz\|_2\leq \alpha \sqrt{n}
	$$ 
	for sufficiently small $\alpha>0$.

	\subsection{Gradient-MUSIC for Fourier subspace estimation} 
	
	Before we discuss how to perform Fourier subspace estimation using Gradient-MUSIC, we start with a general local stability result that is detached from computation. It tells us how much $U(\bfx)$ changes when a well-separated $\bfx$ gets perturbed by a small amount. 
		\begin{proposition} \label{prop:stability1}
			For $\beta \geq 2$ and $n\geq \beta r$, suppose $\bfx,\tilde\bfx\subset \T$ are both sets of cardinality $r$ such that $\Delta(\bfx)\geq 2 \pi \beta/n$ and $\|\bfx-\tilde\bfx\|_\infty \leq \pi \beta/(4n)$. Then there are absolute constants $C_1,C_2>0$ such that
			\begin{align*}
				\big\|\sin(U(\bfx), U(\tilde\bfx)) \big\|_2
				&\leq C_1 n \,  \|\bfx-\tilde\bfx\|_\infty, \\
				\big\|\sin(U(\bfx), U(\tilde\bfx)) \big\|_{\rm F}
				&\leq C_2 n \,  \|\bfx-\tilde\bfx\|_2. 
			\end{align*}
		\end{proposition}
	
	This proposition is proved in \cref{proof:stability1}.
	Under the assumptions that $(\bfx,\bfa)\in \calP(16\pi/n,r)$ and $\|\bfz\|_2\leq \alpha \sqrt{n}$ for a small enough absolute $\alpha>0$, the result in \cite[Lemma 5.6]{fannjiang2025optimality} also holds for Hankel matrices. This implies 
	$$
	\sigma_r(H(\bfy+\bfz))>\sigma_{r+1}(H(\bfy+\bfz)),
	$$
	and that $r$ can be exactly determined by a simple singular value thresholding method described in \cref{rem:threshold}.
	
	Let $W\subset\C^m$ be the leading $r$-dimensional left singular subspace of the Hankel matrix $H(\bfy+\bfz)$. Note that $W$ is not a Fourier subspace, so the truncated SVD does not succeed for Fourier subspace estimation. Feeding $W$ into Gradient-MUSIC, we obtain a set $\bfx^\sharp\subset\T$ of cardinality $r$. Then Gradient-MUSIC produces the Fourier subspace
	$$
	U_\sharp := \range(\Phi_n(\bfx^\sharp))\subset \C^{n}.
	$$
	The following theorem quantifies the quality of approximation. 
	
	\begin{theorem} \label{thm:Fouriersubspace}
		There are absolute constants $C,\alpha>0$ such that the following hold. For all $r\geq 1$ and $n\geq \max\{200,16r\}$, given $\tilde\bfy=\bfy+\bfz$, where $\bfy=\Phi_n(\bfx)\bfa$, $(\bfx,\bfa)\in \calP(16\pi/n,r)$ and $\|\bfz\|_2\leq \alpha \sqrt{n}$, Gradient-MUSIC produces a Fourier subspace $U_\sharp$ that satisfies
		\begin{align*}
			\norm{\sin( U, U_\sharp)}_2
			\leq 
			\norm{\sin( U,U_\sharp)}_{\rm F}
			\leq \frac{C \|\bfz\|_2}{\sqrt n}. 
		\end{align*}
	\end{theorem} 
	
	\begin{proof}
		Recall the definition of $m:=\lceil n/2\rceil$. Let $ U,W\subset\C^m$ denote the leading $r$-dimensional left singular space of $ H(\bfy)$ and $ H(\bfy+\bfz)$ respectively. Note that $ U=\range(\Phi_m(\bfx))\subset\C^m$ due to factorization \eqref{eq:hankelfactorization}, so $ U$ is a Fourier subspace. We see that 
		$$
		\norm{H(\bfz)}_2
		\leq \norm{H(\bfz)}_{\rm F}
		\leq \sqrt m\, \|\bfz\|_2
		\leq \alpha n.
		$$
		Since $(\bfx,\bfa)\in \calP(16\pi/n,r)$, using \eqref{eq:hankelfactorization} and \eqref{eq:Fouriersigs}, we get that
		$$
		\sigma_r(H(\bfy))
		\geq \min_j |a_j| \, \sigma_r^2\round{\Phi_m(\bfx)}
		\gtrsim n.
		$$
		We apply Wedin's sine-theta theorem for Frobenius norm \cite[Chapter V, Theorem 4.1]{stewart1990matrix} to the matrices $ H(\bfy)$ and $ H(\bfy+\bfz)$, and by making $\alpha>0$ sufficiently small, we see that
		\begin{equation}
			\label{eq:sintheta1}
			\norm{\sin( U, W)}_{\rm F}
			\leq \frac{\norm{H(\bfz)}_{\rm F}}{\sigma_r(H(\bfy))-\norm{H(\bfz)}_2}
			\lesssim \frac{\norm{H(\bfz)}_{\rm F}}{n}
			\leq \frac{\|\bfz\|_2}{\sqrt n}
			\leq \alpha.
		\end{equation}
		
		We proceed to verify that the assumptions of \cref{thm:maingradientMUSIC} hold. Observe that $\Delta(\bfx)\geq 16\pi/n \geq 8\pi/m$ and $U$ is a Fourier subspace. Also note that 
		$n\geq 200$ implies $m\geq 100$. Inequality \eqref{eq:sintheta1} tells us that by making $\alpha>0$ a sufficiently small absolute constant, we can assume $\|\sin( U, W)\|_{\rm F}\leq 1/100$. This now verifies that the assumptions of \cref{thm:maingradientMUSIC2} hold, so the Gradient-MUSIC algorithm takes $W$ and outputs a $\bfx^\sharp\subset\T$ of cardinality $r$ such that
		\begin{equation}
			\label{eq:xbound}
			\|\bfx^\sharp-\bfx\|_2
			\lesssim \frac1 m \norm{\sin( U, W)}_{\rm F}
			\lesssim \frac{\|\bfz\|_2}{n^{3/2}}
			\lesssim \frac{\alpha}{n}.
		\end{equation}
		
		Now let $U_\sharp=\range(\Phi_n(\bfx^\sharp))\subset\C^n$ be the Gradient-MUSIC Fourier subspace estimator. By making $\alpha>0$ smaller if necessary, we can make the right hand side of \eqref{eq:xbound} no larger than $4\pi/n$. Since $\Delta(\bfx)\geq 16\pi/n$, we have $\Delta(\bfx^\sharp)\geq 8\pi/n$. Using inequality \eqref{eq:Fouriersigs} again, we see that 
		\begin{equation}
			\label{eq:errorhelp}
			\sigma_r(\Phi_\sharp)
			\gtrsim \sqrt n. 
		\end{equation}
				
		Now we are ready to control the error between $ U$ and $U_\sharp$ by using the Wedin's sine-theta theorem for the Frobenius norm \cite[Chapter V, Theorem 4.1]{stewart1990matrix}. For convenience, let $ E:=\Phi_n(\bfx)-\Phi_n(\bfx^\sharp)$. By \eqref{eq:xbound}, \eqref{eq:errorhelp}, and \cref{lem:Phidifferencespectralnorm},
		$$
		\norm{\sin( U,U_\sharp)}_{\rm F}
		\lesssim \frac{\norm{ E}_{\rm F}}{\sigma_r(\Phi_n(\bfx))-\| E\|_2}
		\lesssim \frac{\| E\|_{\rm F}}{\sqrt n} 
		\lesssim n \|\bfx^\sharp-\bfx\|_2
		\lesssim \frac {\|\bfz\|_2}{\sqrt n}.
		$$
		This proves the claimed bound between $U$ and $U_\sharp$.
	\end{proof}

	Let us examine the computational complexity of the Gradient-MUSIC Fourier subspace estimator. The subsequent discussion shows that computing $U_\sharp$ has worst case computational complexity
	\begin{equation}
		\label{eq:computationalcost2}
		O\left(n^2 r + nr^2 \log\left( \frac {\sqrt n}{\|\bfz\|_2} \right) \right).
	\end{equation}
	The first step is to compute $W$, the leading $r$ dimensional singular space of the Hankel matrix $ H(\bfy)$. This has worst case computational complexity of $O(n^2r)$. Note that $W$ is not guaranteed to be a Fourier subspace, otherwise we would be done. The second step is computing $\bfx^\sharp$ using Gradient-MUSIC. Using the general complexity bound \eqref{eq:computationalcost} and the upper bound  \eqref{eq:sintheta1} the worst case computational and time complexity (with $r\leq p\leq n$ parallel cores) of Gradient-MUSIC are, respectively, 
	\begin{equation*}
		O\left(n^2 r + nr^2 \log\left( \frac {\sqrt n}{\|\bfz\|_2} \right)\right) \andspace
		O\left(\frac{n^2 r}{p} + nr \log\left( \frac {\sqrt n}{\|\bfz\|_2} \right)\right).
	\end{equation*}
	The third and final step is to form $\Phi_\sharp:=\Phi_n(\bfx^\sharp)$ and compute $U_\sharp$ through a QR decomposition, which has complexity at most $O(n^2 r)$. 
	
	In \cite[Lemma 5.3]{fannjiang2025optimality}, we showed that there is a bijection between subsets of $\T$ of cardinality $r$ and the collection of Fourier subspaces of dimension $r$ in $\C^n$, whenever $n\geq 2r$. This is a general statement that does not place assumptions on the separation of $\bfx$. This statement has been known since the inception of MUSIC in \cite{schmidt1986multiple}, though not stated in this form. To see the connection between this statement and MUSIC, it can be shown that whenever $U(\bfx)$ is a Fourier subspace, then its associated MUSIC function $q$ in \eqref{eq:musicfunctions} has exactly $r$ zeros, each of multiplicity two, which are exactly $\bfx=\{x_k\}_{k=1}^r$. In the well-separated case, we can strengthen this qualitative bijection statement to the following local bi-Lipschitz equivalence. 
	\begin{proposition}\label{prop:bilipschitz}
		There are absolute constants $C_1,C_2,\alpha>0$ such that the following hold. Let $r\geq 1$ and $n\geq \max\{200,2r\}$. Suppose $\bfx,\tilde \bfx\subset\T$ are sets of cardinality $r$, $\Delta(\bfx)\geq 16\pi/n$, and $\|\tilde\bfx-\bfx\|_2\leq \alpha/n$. Then 
		\begin{align*}
			\frac n{7} \, \|\tilde\bfx-\bfx\|_\infty
			\leq &\big\|\sin(U(\bfx), U(\tilde\bfx)) \big\|_2
			\leq C_1 n \,  \|\bfx-\tilde\bfx\|_\infty, \\ 
			\frac n{20} \, \|\tilde\bfx-\bfx\|_2
			\leq &\big\|\sin(U(\bfx), U(\tilde\bfx)) \big\|_{\rm F}
			\leq C_2 n \,  \|\bfx-\tilde\bfx\|_2.
		\end{align*}
	\end{proposition}
	
	\begin{proof}
		The two upper bounds immediately follow from \cref{prop:stability1} whenever $\alpha$ is small enough. The lower bounds are essentially the approximation components of \cref{thm:maingradientMUSIC2,thm:maingradientMUSIC} in disguise. We view $U(\bfx)$ as the true Fourier subspace while $U(\tilde \bfx)$ as a perturbed subspace (which also happens to be a Fourier subspace as well). The upper bounds of this proposition together with the assumption that $\|\tilde \bfx-\bfx\|\leq \alpha/n$ imply
		\begin{align*}
			\big\|\sin(U(\bfx), U(\tilde\bfx)) \big\|_{\rm F}
			&\leq C_2 n \,  \|\bfx-\tilde\bfx\|_2
			\leq C_2 \alpha. 
		\end{align*}
		By making $\alpha$ sufficiently small, this term is at most $1/100$ and we can use $U(\tilde\bfx)$ as the input for Gradient-MUSIC. Since $U(\tilde\bfx)$ is itself a Fourier subspace, the MUSIC function $\tilde q$ associated to $U(\tilde\bfx)$ has exactly $r$ zeros, which are precisely $\tilde\bfx$. By \cref{thm:maingradientMUSIC} (see discussion afterwards) and \cref{thm:maingradientMUSIC2}, we see that 
		\begin{align*}
			\|\tilde\bfx-\bfx\|_\infty
			&\leq \frac {7} n \, \big\|\sin(U(\bfx), U(\tilde\bfx)) \big\|_2 \\
			\|\tilde\bfx-\bfx\|_2
			&\leq \frac {20} n \, \big\|\sin(U(\bfx), U(\tilde\bfx)) \big\|_{\rm F}. 
		\end{align*}
		Rearranging this inequality completes the proof. 
	\end{proof}

	\subsection{Lower bound and optimality} 
	
	We proceed to define the minimax error for Fourier subspace estimation. Let $\epsilon>0$ be arbitrary. A method $\psi$ is any function that maps any observation $\bfy+\bfz\in\C^n$ (where $\bfy=\Phi_n(\bfx)\bfa$ for some $(\bfx,\bfa)\in \calP(16\pi/n,r)$ and $\|\bfz\|_2\leq \epsilon$), to a Fourier subspace $ U_\psi$. As usual, we identify $\psi$ with its output $ U_\psi$ and let $ U$ be the Fourier subspace uniquely associated with $\bfx$, as seen in \cref{lem:bijection}. We define the minimax error
	\begin{equation}
		\label{eq:Fminimax}
		F_*(n,r,\epsilon)
		:=\inf_{\psi} \, \,  \sup_{(\bfx,\bfa)\in \calP(16\pi/n,r)} \sup_{\|\bfz\|_2\leq \epsilon} \, \, \norm{\sin( U, U_\psi)}_2,
	\end{equation}
	where the inf is taken over all possible functions $\psi$. 
	
	To obtain a lower bound for $F_*(n,r,\epsilon)$, we will give a proof by contradiction. If there is a method that is too good for Fourier subspace estimation, then combining it with Gradient-MUSIC results in a method that is too good for spectral estimation.
	
	\begin{theorem} \label{thm:Fouriersubspace2}
		For all $r\geq 1$, $n\geq \max\{200,16r\}$, and $\alpha\leq 1$, we have
		\begin{equation*}
			\label{eq:Flower}
			F_*\round{n,r,\alpha  \sqrt n}
			\geq \frac 1{100} \alpha. 
		\end{equation*}
	\end{theorem} 
	
	\begin{proof}
		Let $C=1/100$. Suppose for the purpose of deriving a contradiction that there is a method $\psi$ such that 
		\begin{equation}
			\label{eq:toogood2}
			\sup_{(\bfx,\bfa)\in \calP(16\pi/n,r)} \sup_{\|\bfz\|_2\leq \alpha \sqrt{n}} \, \norm{\sin( U, U_\psi)}_2
			< C \alpha.
		\end{equation}
		The right side of inequality \eqref{eq:toogood2} is at most $1/100$. We feed $ U_\psi$ into Gradient-MUSIC, and \cref{thm:maingradientMUSIC} ensures that this process produces a $\bfx^\psi\subset\T$ of cardinality $r$ such that 
		$$
		\|\bfx-\bfx^\psi\|_\infty
		\leq \frac{10} n \, \norm{\sin( U, U_\psi)}_2
		< \frac{10 C \alpha} n.
		$$
		
		The strategy is to show that this inequality is incompatible with the minimax lower bound for frequency estimation. Consider any $(\bfx,\bfa)\in \calP(16\pi/n,r)$ and $\bfz\in \C^{n}$ such that $\|\bfz\|_2\leq \alpha \sqrt{n}$. Given the observation $\bfy=\Phi_n(\bfx)\bfa+\bfz$, we feed it to method $\psi$, which produces a subspace $ U_\psi\subset \C^n$ of dimension $r$ satisfying inequality \eqref{eq:toogood2}. This holds uniformly, hence
		$$
		\sup_{(\bfx,\bfa)\in \calP(16\pi/n,r)} \,  \sup_{\|\bfz\|_2 \leq \alpha \sqrt{n} } \, \,  \|\bfx-\bfx^\psi\|_\infty
		< \frac{10 C\alpha} n. 
		$$
		Since $C=1/100$, this inequality is incompatible with \cref{lem:minimax}, which yields the desired contradiction.  
	\end{proof}
	
	Now we are ready to derive a two-sided bound for the minimax error which will also show that the Gradient-MUSIC Fourier subspace estimator is optimal in $n$, the noise-to-signal ratio $\alpha^2$, and subspace dimension $r$. 
	
	\begin{corollary}\label{cor:gradmusicfourier}
		There is an absolute constant $C>0$ such that the following hold. For all $r\geq 1$, $n\geq \max\{200,16r\}$, and sufficiently small $\alpha>0$, we have
		\begin{equation*}
			\frac\alpha {100}
			\leq F_*\round{n,r,\alpha \sqrt{n}}
			\leq C\alpha. 
		\end{equation*}
	\end{corollary}
	
	\begin{proof}
		Let $\alpha\leq 1$ be sufficiently small so that it fulfills the hypotheses of \cref{thm:Fouriersubspace} and \cref{thm:Fouriersubspace2}. The lower bound for $F_*$ follows immediately from \cref{thm:Fouriersubspace2}. For the upper bound, since the Gradient-MUSIC Fourier subspace estimator is a method for this problem, by \cref{thm:Fouriersubspace}, we get 
		$$
		F_*\round{n,r,\alpha \sqrt n}
		\leq \sup_{(\bfx,\bfa)\in \calP(16\pi/n,r)} \sup_{\|\bfz\|_2\leq \alpha \sqrt n} \, \norm{\sin( U, U_\sharp)}_2
		\leq C\alpha. 
		$$
		This finishes the proof. 
	\end{proof}
	
	\begin{remark}
		\cref{thm:Fouriersubspace,thm:Fouriersubspace2,cor:gradmusicfourier} hold verbatim if $\calT(n,r)$ is swapped with $\calH(n,r)$, which was defined in \eqref{eq:Hms}. 
	\end{remark}

	\subsection{Discussion}
	\label{sec:discussion1}
	
	We are not aware of any prior work that deals with Fourier subspace estimation. There is a generic optimization approach to this problem, which we explain below and describe its relationship to Gradient-MUSIC. 
	
	Slightly abusing notation, we also think of $\bfx$ as a coordinate $\bfx\in\T^r$. For a given empirical subspace $W\in \Gr(r,\C^n)$, find the best $\bfx\in \T^r$ that minimizes the squared Frobenius distance
		$$
		J(\bfx)
		=
		\frac12\|P(\bfx)-P_{W}\|_{\rm F}^2
		=
		\|\sin(U(\bfx),W)\|_{\rm F}^2, 
		$$
	where $P(\bfx)$ is the projection onto $U(\bfx):=\range(\Phi_n(\bfx))$.

		The objective  $J$ is smooth in the frequency coordinates \(\bfx\) provided that \(\Phi(\bfx):=\Phi_n(\bfx)\) has full
		column rank, which is equivalent to $\bfx$ having $r$ distinct elements. The Frobenius objective has the partial derivative
		\[
		\partial_{x_j}J(\bfx)
		=
		-\operatorname{tr}
		\bigl(\partial_{x_j}P(\bfx)P_W\bigr). 
		\]
		However, in order to minimize $J$ using gradient descent, each step involves  an expensive computation  of the pseudo-inverse $\Phi(\bfx)^\dagger$.

		To simplify computation, since the Vandermonde matrix $\Phi_n(\bfx)$ is well conditioned in the well-separated regime, we modify the objective functional $J$ to $\calJ$, which we call the {\em Grassmannian MUSIC function}, defined as
		\begin{align*}
			\calJ(\bfx) &:=\frac{1}{2n} \left\|P(\bfx)\Phi(\bfx)-P_W\Phi(\bfx)\right\|_{\rm F}^2 \\
			&=\frac{1}{2n} \left\|\Phi(\bfx)-P_W\Phi(\bfx)\right\|_{\rm F}^2 \\
			&=\frac r 2 - \frac{1}{n} \operatorname{tr}\big(\Phi(\bfx)^* P_W\Phi(\bfx)\big) + \frac 1 {2n} \|P_W\Phi(\bfx)\|_{\rm F}^2 \\
			&=\frac{r}{2}-\frac{1}{2n} \operatorname{tr}\big(\Phi(\bfx)^* P_W\Phi(\bfx)\big).\label{eq:grass-music}
		\end{align*}
		Letting $D=\diag(0,i,2i,\dots,i(n-1))$, the gradient of $\calJ$ given by
		\[
\partial_{x_k}\mathcal J(\bfx)
=
-\frac1n\Re\left[\bfphi(x_k)^* P_W\bfphi'(x_k)\right].
\]
		Thus, minimizing $\calJ$ using gradient descent is significantly less expensive than doing the same for $J$ precisely because the same $P_W$ is used in each gradient iteration when minimizing $\calJ$. 
		
		From this point of view, minimizing each $x_k$ separately is analogous to running gradient descent on the MUSIC function associated to $W$, because
		\begin{align*}
			\tilde q(t)&=1-\big\|P_W \bm\phi(t)\big\|_2^2, \\
			\tilde q\, '(t)&= - \frac2n\Re( \bfphi(t)^* P_W \bfphi'(t)). 
		\end{align*}
		As for any gradient-based algorithms, a good initialization strategy is essential for successful performance. This is the key component of the Gradient-MUSIC adopted which  has three major benefits:  simplicity, efficiency and optimal performance guarantee. 
	
	\section{Toeplitz matrix approximation}
	\label{sec:toeplitz}
	
	\subsection{Problem formulation}
	
	We begin by discussing the problem and setting it up in a careful manner. Suppose we observe a matrix
	\begin{equation} \label{eq:Toeplitzmodel}
		 M =  T +  E \in \C^{n\times n}. 
	\end{equation}
	Here, we assume $ T$ is Toeplitz with rank $r<n$ and $E$ is an arbitrary matrix. The correct rank $r$ is unknown. We do not place any structural assumptions on $ E$ (e.g., we do not assume it is Toeplitz) and $ M$ is not necessarily Toeplitz. 
	
	The general difficulties of recovering $ T$ from $ M$ is discussed in \cite{cybenko1982moment,suffridge1993approximation,chu2003structured}. To study this problem in greater depth, we need some assumptions on $ T$. The class of Toeplitz matrices we consider are those with a Fourier matrix factorization whose parameters are well separated. Recall the definition of $\calP$ defined in \eqref{eq:parameterset}. Define the set
	\begin{equation}\label{eq:Tms}
		\calT(n,r)
		:=\left\{ T \in \C^{n\times n}\colon  T = \Phi_n(\bfx) \diag(\bfa) \Phi_n(\bfx)^*, \, (\bfx,\bfa)\in \calP(8\pi/n,r) \right\}.
		%	, \\
		%	&\text{ where } \Delta(\bfx) \geq \frac{8\pi}m, \, |\bfx|=s, \text{ and } 1\leq |a_j|\leq 10 \text{ for each $j=1,\dots,r$} \Bigg\}.
	\end{equation}
	This is a set of Toeplitz matrices which admit a Fourier matrix decomposition where each Fourier matrix is well-conditioned due to \eqref{eq:Fouriersigs}. When necessary to emphasize the parameters associated with $T$, we write $T(\bfx,\bfa)$ as short hand for $\Phi_n(\bfx) \diag(\bfa) \Phi_n(\bfx)^*$.

	Although this set of matrices may look strange, we argue that this is quite natural. 
	\begin{enumerate}[(a)]
		\item 
		The assumption that $1\leq |a_j|\leq 10$ for each $j$ is not restrictive. By rescaling equation \eqref{eq:Toeplitzmodel}, we can assume that $|a_j|\geq 1$, and hence, we are really assuming that $\max_j|a_j|/\min_j |a_j|\leq 10$. The number `10' has no special importance and only affects the other absolute constants that appear in this paper. The dependence of these absolute constants on the condition number of $\diag(\bfa)$ is pursued here.
		\item 
		Due to the Fourier matrix factorization of $ T\in \calT(n,r)$ and inequality \eqref{eq:Fouriersigs}, all $r$ nonzero singular values of $ T$ are in $[cn,Cn]$ for some absolute constants $C>c>0$. On the other hand, if we relax the minimum separation requirement, then $\Phi_n(\bfx)$ may be extremely ill-conditioned due to various results about Fourier matrices, see \cite{li2021stable,batenkov2020conditioning,kunis2020smallest,batenkov2021spectral,barnett2022exponentially,li2025multiscale}. While relaxing the assumption on $\Delta(\bfx)$ would be very interesting, doing so would greatly change the nature of this problem. 
		\item 
		Assuming that $T$ has a Fourier matrix factorization is not prohibitive since it is known that many Toeplitz matrices enjoy such a factorization. Any PSD Toeplitz $ T\in \C^{n\times n}$ with rank $r<n$ can be written as $\Phi_n(\bfx)\diag(\bfa)\Phi_n(\bfx)^*$, where $\bfa$ has positive entries, see \cite[Theorem 1]{cybenko1982moment} and \cite{grenander1958toeplitz,suffridge1993approximation} for related results. Note that a $ T\in \calT(n,r)$ does not have to be PSD since we do not assume that $a_j>0$ for all $j$. There is an analogous factorization result for Hankel matrices \cite[Theorem 3.1.1]{hall1967combinatorial}. 
	\end{enumerate}
	Having completed our justification of the assumptions on $ T$, we can now place suitable requirements on $ E$. Given that all nonzero singular values of $ T$ are proportional to $n$, we need to assume there is a sufficiently small absolute $\alpha>0$ such that 
	\begin{equation}
		\label{eq:noise}
		\| E\|_2\leq \alpha n.
	\end{equation}
	This is the weakest possible assumption on $\| E\|_2$ up to an absolute constant under the above conditions on $ T$. From the viewpoint of signal processing, $\alpha^2$ can be thought of as a noise-to-signal ratio since it controls the strength of $\| E\|_2^2$ versus $\sigma_r^2( T)\asymp n^2$.

	\subsection{Gradient-MUSIC Toeplitz estimator}
	
	Before we explain how to use Gradient-MUSIC for the Toeplitz problem, we start with a local stability result. The following proposition quantifies how much $T(\bfx,\bfa)$ when its parameters are perturbed and is proved in \cref{proof:stability2}.   
	
	\begin{proposition}\label{prop:stability2}
		For $\beta \geq 2$ and $n\geq \beta r$, let $(\bfx,\bfa)\in \calP(2\pi\beta/n,r)$ and $T(\bfx,\bfa)\in \C^{n\times n}$ be its associated Toeplitz matrix. For any $(\tilde \bfx,\tilde \bfa)$ such that $\|\bfx-\tilde \bfx\|_\infty\leq \pi \beta/(2n)$, consider its associated Toeplitz matrix $T(\tilde \bfx,\tilde \bfa)\in \C^{n\times n}$. Then there are absolute constants $C_1,C_2>0$ such that
		\begin{align*}
			\big\|T(\tilde \bfx,\tilde\bfa) -  T(\bfx, \bfa)\big\|_2 &\leq C_1 \left( n^2 \|\tilde\bfx- \bfx\|_\infty + n \big\|\tilde\bfa -  \bfa\big\|_\infty\right) , \\
			\big\|T(\tilde \bfx,\tilde\bfa) -  T(\bfx, \bfa)\big\|_{\rm F} &\leq C_2 \left( n^2 \|\tilde \bfx-\bfx\|_2 + n \big\|\tilde \bfa - \bfa\big\|_2 \right).
		\end{align*}
	\end{proposition}
	
	This proposition makes no assumptions on how $(\tilde\bfx,\tilde\bfa)$ is obtained which makes it completely separate from computation. 
	
	If $\|E\|_2\leq \alpha n$ for a sufficiently small absolute constant $\alpha>0$, the rank $r$ is exactly the number of singular values of $M$ exceeding a fixed absolute threshold $c n$, where $c$ depends only on the amplitude and separation constants, as given in  a thresholding procedure described in \cref{rem:threshold}.

	Let $W$ be the leading $r$-dimensional left singular subspace of $ M$ found by computing the truncated singular value decomposition of $M$. Using $W$ as the input subspace of Gradient-MUSIC, let $\bfx^\sharp \subset\T$ be its output. We define 
	\begin{equation}
		\label{eq:hata}
		\Phi_\sharp:= \Phi_n(\bfx^\sharp), \quad 
		A_\sharp := \Phi_\sharp^\dagger  M \Phi_\sharp ^{\dagger,*}, \andspace
		\bfa^\sharp := \diag( A_\sharp).
	\end{equation}
	Gradient-MUSIC Toeplitz produces the matrix,
	\begin{equation}
		\label{eq:hatT}
		T_\sharp:= T(\bfx^\sharp,\bfa^\sharp) = \Phi_\sharp \, \diag(\bfa^\sharp) \,  \Phi_\sharp^*.
	\end{equation}
	This is clearly a Toeplitz matrix with rank exactly $r$.

	\begin{theorem}\label{thm:gdmusictoeplitz}
		There are absolute constants $C_1, C_2,\alpha >0$ such that the following hold. Suppose $r\geq 1$, $n\geq \max\{100,8r\}$, $ T\in \calT(n,r)$, and $ E\in \C^{n\times n}$ such that $\| E\|_2\leq \alpha n$. Given $ M= T+ E$, Gradient-MUSIC outputs a Toeplitz matrix  $T_\sharp$ of rank $r$ and satisfies
		\begin{align*}
			\big\| T_\sharp -  T\big\|_2 &\leq C_1 \| E\|_2, \\
			\big\| T_\sharp -  T\big\|_{\rm F} &\leq C_2 \min\{ \sqrt{r} \,  \| E\|_2, \|E\|_{\rm F}\}.
		\end{align*}
	\end{theorem}
	
	\begin{proof}		
		Let $W$ be the leading $r$-dimensional left singular subspace of $ M$ found by computing the truncated singular value decomposition of $M$. We control the distance between $W$ and $ U=U(\bfx)$, where the latter is the true $r$-dimensional left singular space of $ T$ and is a Fourier subspace. By Wedin's sine-theta theorem \cite[Chapter V, Theorem 4.4]{stewart1990matrix}, inequality \eqref{eq:Fouriersigs}, and the assumption that $\| E\|_2\leq \alpha n$ for small enough $\alpha>0$, we have
		\begin{equation}
			\label{eq:sinthetatoeplitz}
			\norm{\sin( U, W)}_2
			\leq \frac{\| E\|_2}{\sigma_r( T)-\| E\|_2}
			\lesssim \frac{\| E\|_2}{n}.
		\end{equation}
		Making $\alpha>0$ small enough, this implies $\|\sin( U, W)\|\leq 1/100$. By \cref{thm:maingradientMUSIC}, we obtain 
		\begin{equation}
			\label{eq:gradmusic}
			\|\bfx^\sharp-\bfx\|_\infty
			\lesssim \frac 1n \norm{\sin( U,W)}_2
			\lesssim \frac{\| E\|_2}{n^2}. 
		\end{equation}
		
		Making $\alpha>0$ small enough, we use inequality \eqref{eq:gradmusic} to deduce that $\| \bfx^\sharp-\bfx\|_\infty\leq \pi/n$. Using \cref{lem:amplitudes}, we get
		\begin{equation} 
			\label{eq:help2}
			\big\|\bfa^\sharp - \bfa\big\|_\infty
			\lesssim \frac{\| E\|_2}{n}.
		\end{equation}		
		Inserting inequalities \eqref{eq:gradmusic} and \eqref{eq:help2} into \cref{prop:stability2} completes the proof of the spectral norm estimate.
		
		For the Frobenius norm estimate, since both $T_\sharp$ and $T$ are rank $r$ matrices, $T_\sharp -  T$ has rank at most $2r$. Then
		$$
		\big\| T_\sharp -  T\big\|_{\rm F}\leq \sqrt{2r} \, \big\| T_\sharp -  T\big\|_2
		\lesssim \sqrt r \|E\|_2.
		$$
		Thus, it suffices to prove the first claimed inequality for the spectral norm error. On the other hand, repeating the same steps as earlier except using Wedin's sine-theta theorem for the Frobenius norm \cite[Chapter V, Theorem 4.1]{stewart1990matrix} and the $\ell^2$ Gradient-MUSIC result in \cref{thm:maingradientMUSIC2}, we would instead get
		\begin{align*}
			\|\bfx^\sharp-\bfx\|_2
			&\lesssim \frac 1n \norm{\sin( U,W)}_{\rm F}
			\lesssim \frac{\| E\|_{\rm F}}{n^2}, \\
			\big\|\bfa^\sharp - \bfa\big\|_2
			&\lesssim \frac{\| E\|_{\rm F}}{n}.
		\end{align*}
		Inserting these into the Frobenius norm error bound in \cref{prop:stability2} completes the proof. 
	\end{proof}
	
	The noise condition $\| E\|_2\leq \alpha n$ is not stringent. For deterministic matrix $ E$, the requirement that its spectral norm scales linearly with $n$ is reasonable. This condition is also satisfied for various random matrices. For example, if $ E$ is a Toeplitz matrix generated by a sequence $e_{-n+1},\dots,e_{n-1}$ of i.i.d. normal random variables with mean zero and variance $\sigma^2$, due to \cite[Theorem 4.1.1]{tropp2012user}, with probability $1-o(1)$, it holds that $\| E\|_2\lesssim \sigma\sqrt{n\log(n)}$. Hence $ E$ satisfies \cref{thm:gdmusictoeplitz} for any $\sigma$ provided that $n$ is large enough. If $ E$ is a $n\times n$ matrix with i.i.d. normal random variables, then it satisfies the same spectral norm upper bound, see \cite[Theorem 4.4.5]{vershynin2018high}.

	Let us examine the computational complexity of the Gradient-MUSIC Toeplitz estimator. The subsequent discussion shows that computing $T_\sharp$ has worst case computational complexity
	\begin{equation}
		\label{eq:computationalcost1}
		O\left(n^2 r + nr^2 \log\left( \frac {n}{\| E\|_2} \right)\right).
	\end{equation}
	The first step is to compute $W$, the leading $r$ dimensional singular space of $ M$. This has worst case computational complexity of $O(n^2r)$. The second step is computing $\bfx^\sharp$ using Gradient-MUSIC. Using the general complexity bound \eqref{eq:computationalcost} and the upper bound \eqref{eq:sinthetatoeplitz}, the worst case computational and time complexity (with $r\leq p\leq n$ parallel cores) of Gradient-MUSIC are, respectively, 
	$$
	O\left(n^2 r + nr^2 \log\left( \frac {n}{\| E\|_2} \right)\right) \andspace
	O\left(\frac{n^2 r}{p} + nr \log\left( \frac {n}{\| E\|_2} \right)\right).
	$$
	The third and final step is computing $\bfa^\sharp$ by formula \eqref{eq:hata} and forming $T_\sharp$ by equation \eqref{eq:hatT}. Computing the pseudo inverse of $\Phi_\sharp$ and doing the necessary matrix products to get $\bfa^\sharp$ ultimately has complexity $O(n^2 r)$. Forming $T_\sharp$ has the same order.  
	
	%\noteM{Revisiting our previous discussion: if we evaluate the landscape function on a uniform coarse grid of width $1/(2n)$, then this evaluation can be done with the FFT. This allows us to replace $n^2r$ with $n\log(n)r$ though it cannot be parallelized to the same degree. If we recover $\bfa^\sharp=\Phi^+\bfy$ using least squares instead, this has complexity $O(nr+r^3)$ and has the same theoretical guarantees (for deterministic noise) as the quadratic method. So the question is whether we can reduce the SVD computation to $n\log(n)$ deterministically, which feels like it should be possible... or we can go with your randomized approach.}

	\subsection{Lower bound and optimality} 
	Now we argue that the performance guarantee for Gradient-MUSIC for Toeplitz matrix estimation is minimax optimal in $\| E\|_2$, $n$, and $r$. We start with a few definitions. Fix a $\epsilon>0$. We let $\psi$ be any function such that given any $ T +  E$ with $ T\in \calT(n,r)$ and $\| E\|_2\leq \epsilon$, it outputs a Toeplitz matrix $ T_\psi$ of rank $r$. Here, $\psi$ is a function which only depends on $ T+ E$, and not on either $ T$ or $ E$ individually. Moreover, $\psi$ does not need to be a computable function or implementable by a computationally tractable algorithm. Define the minimax error,
	\begin{equation}
		T_*(n,r,\epsilon)
		= \inf_{\psi} \,  \sup_{ T\in \calT(n,r)} \,  \sup_{\| E\|_2\leq \epsilon} \big\| T_\psi- T\big\|_2,
	\end{equation}
	where the $\inf_{\psi}$ is taken over all possible functions of the observation $ T +  E$. 
	
	We seek a lower bound for $ T_*(n,r,\epsilon)$. The lower bound will be derived via a proof by contradiction. Roughly speaking, if there is a method that is too good for the Toeplitz matrix problem, then combining it with Gradient-MUSIC yields a method that is too good for spectral estimation. 
	
	\begin{theorem} \label{thm:minimaxtoeplitz}
	
		For any $r\geq 1$, $n\geq \max\{100,8r\}$, and $\alpha \leq 1/8$, we have
		$$
		T_*(n,r,\alpha n)\geq \frac 1 {320}\alpha n. 
		$$
	\end{theorem}
	
	\begin{proof}
		Let $C=1/320$. Suppose for the purpose of deriving a contradiction, that there is a method $\psi$ such that
		\begin{equation}
			\label{eq:toogood}
			\sup_{ T\in \calT(n,r)} \, \sup_{\| E\|_2\leq \alpha n} \| T_\psi-  T\|_2
			< C \alpha n.			
		\end{equation}
		
		Let $ U_\psi$ denote the leading $r$-dimensional left singular subspace of $ T_\psi$. Here, we think of $ U_\psi$ as a map from $ T+ E$ to the subspace $ U_\psi$. We need to briefly argue that this subspace is well-defined. By Weyl's inequality for singular values and inequality \eqref{eq:Fouriersigs},  
		\begin{align*}
			\sigma_r(T_\psi)
			&\geq \sigma_r( T)-\| T-T_\psi\|_2
			\geq \frac 3 4n - C\alpha n, \\
			\sigma_{r+1}( T_\psi)
			&\leq \sigma_{r+1}( T)+\|T-T_\psi\|_2
			=\|T-T_\psi\|_2
			\leq C\alpha n.
		\end{align*}
		Since $\alpha \leq 1/8$ by assumption, we see that $\sigma_r( T_\psi)>\sigma_{r+1}( T_\psi)$, which makes the subspace $ U_\psi$ well-defined. 
		
		Moreover, by the Wedin's sine-theta theorem \cite[Chapter V, Theorem 4.4]{stewart1990matrix}, inequality \eqref{eq:toogood}, and the above derived inequalities, we have
		$$
		\norm{\sin( U, U_\psi)}_2
		\leq \frac{\| T_\psi- T\|_2}{\sigma_r( T)-\| T-T_\psi\|_2}
		< \frac{C \alpha n}{3n/4-C\alpha n}
		\leq 2C  \alpha.
		$$
		Note the right hand side is less than $1/100$. Suppose $ U_\psi$ is used in the Gradient-MUSIC algorithm, which produces a $\bfx^\psi$. By \cref{thm:maingradientMUSIC}, we see that 
		\begin{equation}
			\label{eq:toogood3}
			\|\bfx^\psi-\bfx\|_\infty
			\leq \frac{10}n \, \|\sin( U, U_\psi)\|_2
			< \frac{20 C \alpha}n.
		\end{equation}
		
		To finish the proof, we now argue that this inequality cannot hold by relating this to spectral estimation. For each $(\bfx,\bfa)\in \calP(8\pi/n,r)$ and $\bfz\in \C^{2n-1}$ such that $\|\bfz\|_\infty \leq \alpha$, consider the data $\tilde \bfy=\bfy + \bfz=\Phi_{2n-1}(\bfx)\bfa+ \bfz$. The $n\times n$ Toeplitz matrix associated with $\bfy$ is 
		$$
		T(\bfy)
		=\begin{bmatrix}
			y_0 &y_{-1} &\cdots &y_{-n+1} \\
			y_1 &y_0 &\cdots &y_{-n+2} \\
			\vdots &\vdots & &\vdots \\
			y_{n-1} &y_{n-2} &\cdots &y_0
		\end{bmatrix}.
		$$		
		Likewise, let $ T(\bfz)$ be $n\times n$ Toeplitz matrix associated with $\bfz$. We have
		$$
		T(\tilde \bfy)
		=\Phi_n(\bfx)\diag(\bfa)\Phi_n(\bfx)^* + T(\bfz).
		$$
		Notice that
		\begin{equation*}
			\|T(\bfz)\|_2
			\leq \|T(\bfz)\|_{\rm F}
			\leq n \|\bfz\|_\infty
			\leq \alpha n.
		\end{equation*}
		The matrix $ T(\tilde \bfy)$ has the correct form of being a Toeplitz matrix in $\calT(n,r)$ plus an error term whose spectral norm is at most $\alpha n$. Specializing the above method acting on $ T(\tilde \bfy)$ now, by inequality \eqref{eq:toogood3}, we deduce
		$$
		\sup_{(\bfx,\bfa)\in \calP(8\pi/n,r)} \,  \sup_{\|\bfz\|_\infty \leq \alpha} \|\bfx^\psi-\bfx\|_\infty
		<\frac{20 C \alpha}n. 
		$$
		Since $N:=2n-1$ is the total number of measurements and $\Delta(\bfx)\geq 16\pi/N$ implies $\Delta(\bfx)\geq 8\pi/n$, we now see that
		\begin{equation}
			\label{eq:minmax}
			\sup_{(\bfx,\bfa)\in \calP(16\pi/N,r)} \,  \sup_{\|\bfz\|_\infty \leq \alpha} \|\bfx^\psi-\bfx\|_\infty
			<\frac{20 C \alpha}n
			\leq \frac{40 C\alpha}N. 
		\end{equation} 
		Recall the minimax error for spectral estimation in \eqref{eq:minimaxspectral} and a lower bound stated in \cref{lem:minimax}. Since $C=1/320$, assertion \eqref{eq:minmax} and \cref{lem:minimax} are incompatible, which yields the desired contradiction. 
	\end{proof}
	
	Putting the two previous theorems together, we derive a two-sided bound for the minimax error and show that Gradient-MUSIC is optimal in the matrix dimension $n$, noise-to-signal ratio $\alpha^2$, and $r$. 
	
	\begin{corollary}\label{cor:toeplitzoptimal}
		There is an absolute constant $C>0$ such that the following hold. For any $r\geq 1$, $n\geq \max\{100,8r\}$, and sufficiently small $\alpha>0$, we have 
		$$
		\frac 1 {320} \alpha n 
		\leq T_*(n,r, \alpha n)
		\leq C \alpha n. 
		$$
	\end{corollary}
	
	\begin{proof}
		Pick any $\alpha$ sufficiently small so that it satisfies the requirements in \cref{thm:gdmusictoeplitz,thm:minimaxtoeplitz}. The lower bound for $ T_*(n,r,\alpha n)$ follows immediately from \cref{thm:minimaxtoeplitz}. For the upper bound, since $ T_*$ is taken over the inf of all functions (which includes the Gradient-MUSIC Toeplitz estimator $T_\sharp$), we use \cref{thm:gdmusictoeplitz} to see that 
		$$
		T_*(n,r,\alpha n)
		\leq  \sup_{ T\in \calT(n,r)} \,  \sup_{\| E\|_2\leq \alpha n} \big\|T_\sharp -  T\big\|_2
		\leq C \alpha n. 
		$$
		This yields the desired upper bound on $ T_*$.
	\end{proof}
	
	\begin{remark}
		For Hankel matrices, our algorithm produces a rank-$r$ Hankel matrix
		$$
		H_\sharp = \Phi_n(\bfx^\sharp) \diag(\bfa^\sharp) \Phi_n(\bfx^\sharp)^T.  
		$$
		\cref{thm:gdmusictoeplitz} holds verbatim if $\calT(n,r)$ is swapped with the following class of rank deficient Hankel matrices,
		\begin{equation}\label{eq:Hms}
			\calH(n,r)=\left\{ H \in \C^{n\times n}\colon  H = \Phi_n(\bfx) \diag(\bfa) \Phi_n(\bfx)^T, \, (\bfx,\bfa)\in \calP(8\pi/n,r) \right\}.
		\end{equation}
		Likewise, \cref{thm:minimaxtoeplitz,cor:toeplitzoptimal} hold verbatim if $ T_*(m,r,\epsilon)$ is replaced with 
		$$
		H_*(n,r,\epsilon)
		= \inf_{\psi} \,  \sup_{ H\in \calH(n,r)} \,  \sup_{\| E\|_2\leq \epsilon} \big\| H_\psi-  H\big\|_2.
		$$
	\end{remark}
	
	\subsection{Discussion}
	\label{sec:discussion2}	
	
	For Toeplitz and Hankel matrices, the central difficulty is that the low-rank constraint and the
	Toeplitz/Hankel constraint are individually simple but jointly nonlinear and nonconvex.  
	The papers \cite{cadzow1988signal,grigoriadis2002application,chu2003structured} used alternating projection for the Toeplitz approximation problem. The method attempts to find a minimizer by alternating between ``projections" onto the set of rank $r$ matrices and Toeplitz matrices -- note that the set of rank $r$ matrices is not closed and convex so there may not be a well-defined projection. Unless alternating projection converges to a local minimum, which is usually not the case, it will produce a Toeplitz matrix that does not usually have rank $r$. Our numerical simulations in \cref{sec:numerics} demonstrate that it produces a Toeplitz matrix that has large numerical rank. There are general theoretical guarantees for convergence of alternating projection that depends on the geometry of the two sets, \cite{andersson2013alternating,drusvyatskiy2015transversality}.
	
	The papers \cite{rosen1996total,park1999low} proposed a structured total least norm (STLN) method. As explained in \cite[Section 4]{rosen1996total}, the objective function to be minimized is nonconvex and their proposed iterative method is essentially a Newton’s method. Numerical simulations in \cite{park1999low} show that for certain problems, STLN has similar approximation capabilities as Cadzow's method \cite{cadzow1988signal}, where the latter is a type of alternating projection. The computational cost of STLN is significant. Each iteration requires finding a solution to a least squares system with a square $n(n-r)\times n(n-r)$ matrix, see \cite[Section 2.2]{park1999low}. We are not aware of any theoretical results about whether STLN converges or the approximation quality of the computed matrix. 
	
	A completely different approach based on (inverse) sparse Fourier transforms was developed in \cite{kapralov2023toeplitz,musco2024sublinear}. They assume $ T$ is real PSD but do not place any other restrictions on $ T$. To provide a comparison with our work, consider any real PSD $ T\in \calT(n,r)$. Then \cite[Theorem 1]{musco2024sublinear} says that, given a parameter $\delta>0$, there is a random algorithm that outputs a Toeplitz $ T_\flat$ whose rank is $\poly(r,\log(n/\delta))$ such that with probability at least 0.9, 
	\begin{equation}
		\label{eq:musco}
		\norm{ T- T_\flat}_{\rm F} 
		\lesssim \| E\|_{\rm F} + \delta \, \norm{ T}_{\rm F}.
	\end{equation}
	In terms of theoretical guarantee, this result is the strongest that we could find. If one selects $\delta=\| E\|_{\rm F}/\| T\|_{\rm F}$, then the right side of \eqref{eq:musco} becomes $C\| E\|_{\rm F}$, whereas \cref{thm:gdmusictoeplitz} produces $T^\sharp$ with 
	$$
	\|T-T_\sharp\|_{\rm F}\lesssim  \min\{\sqrt r\| E\|_2, \|E\|_{\rm F}\}. 
	$$
	(though our theorem also bounds the error in spectral norm too). The main difference is that $ T_\flat$ does not necessarily have the correct rank of $r$ because their method over estimates the number of frequencies in the noisy signal associated with $ T$. This is worth elaborating on: if one uses a spectral estimation algorithm to create an approximating Toeplitz matrix, overestimating the number of frequencies reduces the approximation error but at the price of overestimating the rank.
	\commentout{
		There is a completely different approach is based on sparse inverse Fourier transforms \cite{PriceSong2015RobustSparseFTContinuous,kapralov2023toeplitz,musco2024sublinear}. While these methods are fast, have theoretical guarantees, and output a Toeplitz matrix, the returned matrix does not always have the correct rank. A more detailed comparison to prior work is deferred to \cref{sec:relatedworks}.  There are also several iterative algorithms based on minimizing a data fidelity term \cite{cybenko1982moment,rosen1996total,park1999low}. 
	}

	Comparing computational complexities, their random method, using $\delta= \| E\|_{\rm F}/\| T\|_{\rm F}$, outputs the parameters defining $ T^\sharp$ with time complexity of $\poly(r,\log(n\norm{ T}_{\rm F}/\| E\|_{\rm F}))$ for an unspecified polynomial. Gradient-MUSIC is a deterministic method and its time complexity is given in \eqref{eq:computationalcost1}. The two upper bounds are not directly comparable, especially since $r$ can be on the order of $n$.  
	
	There are other works that use sparse Fourier transforms to perform spectral estimation. For example, \cite{PriceSong2015RobustSparseFTContinuous} requires random samples on $[0,n]$ and cannot directly be used for approximation by low-rank Toeplitz matrices, which requires integer samples (this is one of the main technical differences between \cite{PriceSong2015RobustSparseFTContinuous} and \cite{musco2024sublinear} which both use sparse Fourier transform methods). The biggest difference between \cite{PriceSong2015RobustSparseFTContinuous} and our Gradient-MUSIC algorithm \cite{fannjiang2025optimality} is the separation condition: the former succeeds when $\Delta \gtrsim \log(r/\delta)/n$ where $n\delta$ is roughly the squared $\ell^2$ norm of the noise, whereas Gradient-MUSIC succeeds when $\Delta\geq 8\pi/n$. In compensation, their random algorithm is faster. 
	
%	Our approach
%	takes a different route: first estimate the approximate signal subspace, then use spectral estimation
%	to recover the hidden Vandermonde parameters, and finally rebuild the structured matrix.
	
	The biggest difference between our work and prior methods is that, under our framework, the Gradient-MUSIC Toeplitz estimator returns a Toeplitz matrix that has the correct rank and optimally accurate Fourier factorization. Other methods do not necessarily output the correct rank and/or do not have theoretical guarantees.  
	
	\section{Exponential sum approximation} 
	\label{sec:exponentialsums}
	
	\subsection{Problem formulation} 
	
	Consider the following set of exponential sums with $r$ nonzero terms,
	$$
	\calE(h,r):=\left\{f \colon \R\to \C \colon f(t) = \sum_{k=1}^r a_k e^{i x_k t}, \,  (\bfx,\bfa)\in \calP(h,r) \right\}. 
	$$
	Here $\bfx$ represents the frequencies and $\bfa$ are the amplitudes. Again $h>0$ lower bounds the distance between frequencies and $r$ is the number of exponentials in $f$. The set of all exponential sums defined on $\R$ is $\bigcup_{h>0}\bigcup_{r\geq 1} \calE(h,r)$.
	
	For a given integer $n\geq 1$, a structured approximation problem is to approximate $f\in \calE(h,r)$ by another exponential sum $\tilde f$ with respect to the error 
	\begin{equation}
		\big\| \tilde f -f \big\|_{L^2([0,n])}. \label{eq:approximationerror}
	\end{equation}
	The error is taken over a bounded interval $[0, n]$, otherwise it is generally possible that $\|\tilde f-f\|_{L^2(\R)}=\infty$ whenever $\tilde f\not=f$. 
	
	We consider a natural setting where we are given noisy samples of $f\in \calE(2\pi\beta/n,r)$ for large enough $\beta > 1$ and on the natural sampling set $\{0,1,\dots,n-1\}$, which we can again denote by $\tilde \bfy$ so that 
	$$
	\bfy = \{f(j)\}_{j=0,\dots,n-1}, \andspace
	\tilde \bfy = \bfy + \bfz,
	$$
	where the unknown additive noise is represented by $\bfz$. In the language of signal processing, attempting to minimize \eqref{eq:approximationerror} from samples can be viewed as a ``denoising" problem. Under the requirement that $f\in \calE(2\pi \beta/n,r)$ for $\beta\geq 2$, we see that $\|\bfy\|_2\gtrsim \sqrt n$ due to \eqref{eq:Fouriersigs}, which means that one must assume there is a sufficiently small absolute $\alpha>0$ such that
	$$
	\|\bfz\|_2\leq \alpha \sqrt n.
	$$
	
	Let us discuss the main difficulties of this problem. The constraint that $\tilde f$ must be an exponential sum is difficult to enforce directly. The error \eqref{eq:approximationerror} can be made smaller by allowing $\tilde f$ to have more than $r$ complex exponentials. This of course naturally raises the more specific question of how to find a $\tilde f$ with exactly $r$ terms that achieves small error. In the language of approximation theory, this would be ``approximation from within" since  we would like to approximate $f$ with a function $\tilde f$ also within the class of exponential sums with exactly $r$ frequencies. 
	
	\subsection{Gradient-MUSIC for approximation of exponential sums}
	
	We recall a classical result about the $L^2([0,n])$ norm of an exponential sum: for any $f\in \calE(h,r)$ with associated parameters $(\bfx,\bfa)$, it holds that
	\begin{equation}
		\left(n-2\pi h^{-1} \right) \|\bfa\|_2^2 
		\leq \int_0^{n} |f(t)|^2 \, dt 
		\leq \left(n+2\pi h^{-1}\right) \|\bfa\|_2^2. \label{eq:beurling}
	\end{equation}
	A proof can be found in \cite[Section 1]{vaaler1985some}, which relies on properties of the Beurling-Selberg minorant and majorant functions. One notices that the lower bound in \eqref{eq:beurling} is vacuous once $h\leq 2\pi /n$. In the special case where $h=2\pi \beta/n$ for a $\beta \geq 2$, we see that \eqref{eq:beurling} reduces to
	\begin{equation} \label{eq:beurling2}
		\int_0^n |f(t)|^2 \, dt 
		\asymp n \|\bfa\|_2^2. 
	\end{equation} 
	This inequality states that the $r$ exponentials in $f$ behave almost orthogonal to each other on the interval $[0,n]$. In spirit, inequality \eqref{eq:Fouriersigs} is \eqref{eq:beurling} together with a Poisson summation formula argument. 
	
	At this point we are ready to prove a local stability estimate for exponential sums in terms of the difference between their parameters. There are two natural choices of norms to use. The first being $L^2([0,n])$. The other is a discrete analogue, 
	$$
	\|f\|_{\ell^2_n}^2 := \sum_{j=0}^{n-1} |f(j)|^2.
	$$
	We start with a stability result which quantifies how a perturbation of $(\bfx,\bfa)$ for well-separated $\bfx$ affects its associated exponential sum. The proposition is proved in \cref{proof:stability3}.
	
	\begin{proposition} \label{prop:stability3}
		Let $\beta\geq 2$ and $n\geq \beta r$. Suppose $f\in \calE(2\pi \beta/n,r)$ with associated parameters $(\bfx,\bfa)$. Suppose $\tilde f$ is an exponential sum with associated parameters $(\tilde \bfx, \tilde \bfa)$ such that $\tilde \bfx$ has cardinality $r$ as well and $\|\tilde \bfx-\bfx\|_\infty \leq \pi \beta /(2n)$. Then there are absolute constants $C_1,C_2>0$ such that 
		\begin{align*}
		\big\|\tilde f - f \big\|_{L^2([0,n])}
		&\leq C_1\left( \sqrt n \, \|\tilde \bfa- \bfa\|_2 + n^{3/2} \|\bfa\|_\infty \, \|\tilde \bfx-\bfx\|_2 \right), \\ 
		\big\| \tilde f- f\big\|_{\ell^2_n}
		&\leq C_2 \left( \sqrt n \, \|\tilde \bfa- \bfa\|_2 + n^{3/2} \|\bfa\|_\infty \, \|\tilde \bfx-\bfx\|_2 \right).
	\end{align*}
	\end{proposition}

	This proposition is purely a stability result which controls the error between two exponential sums in terms of their associated parameters. It illustrates a clear strategy -- any computational method that enables one to approximate an exponential sum's parameters from data $\tilde\bfy$ automatically yields an approximation result.
	
	To compute a suitable $\tilde f$ from noisy data $\tilde \bfy\in \C^n$ generated by an exponential sum $f\in \calE(16\pi/n,r)$, we set $m:=\lceil n/2 \rceil$ and form the square Hankel matrix $H(\tilde\bfy)\in \C^{m\times m}$ as defined in \eqref{eq:hankel}. 
	If $\|\bfz\|_2\leq \alpha \sqrt n$ for a sufficiently small absolute constant $\alpha>0$, then $\|H(\bfz)\|_2\lesssim \alpha m$, so the rank of $H(\tilde\bfy)$ can be deduced from a thresholding procedure as described in \cref{rem:threshold}.
	
	Let $W$ be the leading $r$-dimensional left singular subspace of $H(\tilde\bfy)$. Using $W$ as the input subspace of Gradient-MUSIC, let $\bfx^\sharp \subset\T$ be its output. Then we define 
	\begin{equation}
		\label{eq:hata2}
		\Phi_\sharp:= \Phi_m(\bfx^\sharp), \quad 
		A_\sharp := \Phi_\sharp^\dagger  H(\tilde\bfy) \Phi_\sharp ^{\dagger, T}, \andspace
		\bfa^\sharp := \diag( A_\sharp).
	\end{equation}
	Then {\it Gradient-MUSIC} produces the following exponential sum, 
	\begin{equation}
		\label{eq:hatf}
		f_\sharp(t):= \sum_{k=1}^r a^\sharp_k e^{ix_k^\sharp t}.
	\end{equation}
	We have the following result for the performance of Gradient-MUSIC for approximation of exponential sums from data. 
	
	\begin{theorem} \label{thm:gradmusicexponentialsum}
		There are absolute constants $C_1, C_2,\alpha >0$ such that the following hold. Suppose $r\geq 1$, $n\geq \max\{200,16r\}$, $f\in \calE(16\pi/n,r)$, and $\bfz\in \C^n$ such that $\| \bfz\|_2\leq \alpha \sqrt n$. Given $\tilde \bfy= \bfy+ \bfz$, Gradient-MUSIC outputs an exponential sum $f_\sharp$ with $r$ terms and satisfies
		\begin{align}
			\big\| f_\sharp -  f\big\|_{L^2([0,n])} 
			&\leq C_1 \|\bfz\|_2, \nonumber\\
			\big\| f_\sharp -  f\|_{\ell^2_n}
			&\leq C_2 \|\bfz\|_2.\label{eq:gradmusicexponentialdiscrete}
		\end{align}
	\end{theorem}
	
	\begin{proof}
		Let $(\bfx,\bfa)$ be the parameters associated with $f\in \calE(16\pi/n,r)$. By definition of $\calE(h,r)$ and $\calP(h,r)$, we have $1\leq |a_j|\leq 10$ for all $j$ and since $m:=\lceil n/2\rceil$, we see that
		$$
		\Delta(\bfx)\geq \frac{16 \pi}n \geq \frac{8\pi}m. 
		$$
		Let $W$ be the leading $r$-dimensional left singular subspace of $H(\tilde\bfy)$. We control the distance between $W$ and $U=U(\bfx)\subset \C^m$, where $U$ is the true $r$-dimensional left singular space of $H(\bfy)$ and is a Fourier subspace. Since $\|\bfz\|\leq \alpha \sqrt n$ and $m=\lceil n/2\rceil\geq 100$, we have
		\begin{equation}\label{eq:hankelnoise}
			\|H(\bfz)\|_2 
			\leq \|H(\bfz)\|_{\rm F}
			\leq \sqrt m \|\bfz\|_2
			\lesssim \alpha m. 
		\end{equation}		
		By Wedin's theorem for the Frobenius norm \cite[Chapter V, Theorem 4.1]{stewart1990matrix}, Fourier matrix factorization of Hankel matrix, and inequality \eqref{eq:Fouriersigs},
		\begin{equation}
			\label{eq:sinthetahankel}
			\norm{\sin( U, W)}_{\rm F}
			\leq \frac{\|H(\bfz)\|_{\rm F}}{\sigma_r(H(\bfy))-\|H(\bfz)\|_2}
			\lesssim \frac{\| H(\bfz)\|_{\rm F}}{m}.
		\end{equation}
		Using \eqref{eq:hankelnoise} and making $\alpha>0$ a sufficiently small absolute constant, this implies $\|\sin( U, W)\|_{\rm F}\leq 1/100$. The assumptions of  \cref{thm:maingradientMUSIC2} are fulfilled, so together with \eqref{eq:sinthetahankel}, we obtain 
		\begin{equation}
			\label{eq:gradmusic2}
			\|\bfx^\sharp-\bfx\|_2
			\lesssim \frac 1 m \, \norm{\sin(U,W)}_{\rm F}
			\lesssim \frac{\| H(\bfz)\|_{\rm F}}{m^2}. 
		\end{equation}
		Using \eqref{eq:hankelnoise} and making $\alpha>0$ even smaller if necessary, we use inequality \eqref{eq:gradmusic2} to deduce that $\| \bfx^\sharp-\bfx\|_\infty\leq \pi/m$. It follows from \cref{lem:amplitudes} that 
		\begin{equation} 
			\label{eq:amplitudes2}
			\big\|\bfa^\sharp - \bfa\big\|_2
			\lesssim \frac{\|H(\bfz)\|_{\rm F}}{m}.
		\end{equation}			
		Inserting inequalities \eqref{eq:gradmusic2} and \eqref{eq:amplitudes2} into \cref{prop:stability3} completes the proof.
	\end{proof}
	
	Let us examine the computational complexity of the Gradient-MUSIC for computing the parameters of $f_\sharp$. The subsequent discussion shows that computing $f_\sharp$ has worst case computational complexity
	\begin{equation*}
		O\left(n^2 r + nr^2 \log\left( \frac {\sqrt n}{\|\bfz\|_2} \right)\right).
	\end{equation*}
	The first step is to compute $W$, the leading $r$ dimensional singular space of $H(\tilde\bfy)$. This has worst case computational complexity of $O(n^2r)$. The second step is computing $\bfx^\sharp$ using Gradient-MUSIC. Using the general complexity bound discussed after the statement of \cref{thm:maingradientMUSIC2} and the upper bounds \eqref{eq:hankelnoise} and \eqref{eq:sinthetahankel}, the worst case computational and time complexity (with $r\leq p\leq n$ parallel cores) of Gradient-MUSIC are, respectively, 
	$$
	O\left(n^2 r + nr^2 \log\left( \frac {\sqrt n}{\|\bfz\|_2} \right)\right) \andspace
	O\left(\frac{n^2 r}{p} + nr \log\left( \frac {\sqrt n}{\|\bfz\|_2} \right)\right).
	$$
	The third and final step is computing $\bfa^\sharp$ by formula \eqref{eq:hata2}. Computing the pseudo-inverse of $\Phi_\sharp$ and doing the necessary matrix products to get $\bfa^\sharp$ ultimately has complexity $O(n^2 r)$. Evaluation of $f_\sharp(t)$ at any $t\in \R$ has worst case complexity $O(r)$. 
	
	\subsection{Lower bound and optimality}
	
	In this section, we argue that the approximation rate \eqref{eq:gradmusicexponentialdiscrete} in \cref{thm:gradmusicexponentialsum} is minimax optimal in $r$, $n$ and $\|\bfz\|_2$. We first start with some definitions. Fix $\epsilon>0$. We let $\psi$ be any function such that given any $\bfy + \bfz$ with $\bfy:=f|_{\{0,1,\dots,n-1\}}$, $f\in \calE(16\pi/n,r)$ and $\|\bfz\|_2\leq \epsilon$, it outputs an exponential sum $ f_\psi$ with exactly $r$ exponentials. Here, $\psi$ is a function which only depends on $\tilde\bfy$, and not on either $f|_{\{0,1,\dots,n-1\}} $ or $\bfz$ individually. Moreover, $\psi$ does not need to be a computable function or implementable by a computationally tractable algorithm. Define the minimax error,
	\begin{equation}
		E_*(n,r,\epsilon)
		= \inf_{\psi} \,  \sup_{ f\in \calE(16\pi/n,r)} \,  \sup_{\| \bfz\|_2\leq \epsilon} \big\| f_\psi- f\big\|_{\ell^2_n},
	\end{equation}
	where the $\inf_{\psi}$ is taken over all possible functions of the observation $\bfy + \bfz$. 
	
	\begin{theorem} \label{thm:minimaxexponential}
		For any $r\geq 1$, $n\geq \max\{200,16r\}$, and $\alpha\leq 1$, we have 
		$$
		E_*(n,r,\alpha \sqrt n)\geq \frac 1 {400}\alpha \sqrt n. 
		$$
	\end{theorem}
	
	\begin{proof}
		Let $C=1/400$. Suppose for the purpose of deriving a contradiction, that there is a method $\psi$ such that
		\begin{equation}
			\label{eq:toogood4}
			\sup_{f\in \calE(16\pi/n,r)} \, \sup_{\| \bfz\|_2\leq \alpha \sqrt n} \| f_\psi-  f\|_{\ell_n^2}
			< C \alpha \sqrt n.
		\end{equation}
		Let $\bfy:=f|_{\{0,1,\dots,n-1\}}$ and $\bfy^\psi:=f_\psi|_{\{0,1,\dots,n-1\}}$. Let $(\bfx,\bfa)\in \calP(16\pi/n,r)$ be the parameters associated with $f$ and $(\bfx^\psi,\bfa^\psi)$ be the parameters associated with $f_\psi$. 
		
		Set $m:=\lceil n/2\rceil$. Importantly, since $f_\psi$ is an exponential sum with $r$ exponentials, by the Fourier matrix factorization \eqref{eq:hankelfactorization}, the Hankel matrix $H_\psi:=H(\bfy_\psi)\in \C^{m\times m}$ has rank $r$. Additionally, the $r$-dimensional left singular subspace $U_\psi\subset\C^m$ of $H_\psi$ is itself a Fourier subspace in $\C^m$. 
		
		Next, we control the distance between $U$ and $U_\psi$. For convenience, set $E:=H(\bfy^\psi)-H(\bfy)$. We first note that \eqref{eq:toogood4} yields
		\begin{equation}
			\| E\|_2
			\leq \| E\|_{\rm F}
			\leq \sqrt n \, \|\bfy^\psi-\bfy\|_2
			= \sqrt n \, \| f_\psi-  f\|_{\ell_n^2}
			< C\alpha n.			
			\label{eq:hankelhelp1}
		\end{equation}
		Note that $\Delta(\bfx)\geq 16\pi/n$ implies $\Delta(\bfx)\geq 8\pi /m$. The Hankel matrix factorization \eqref{eq:hankelfactorization} and inequality \eqref{eq:Fouriersigs} also yield
		\begin{equation}
			\sigma_r(H(\bfy))
			\geq \sigma_r^2(\Phi_m(\bfx))
			\geq \frac 34 m 
			\geq \frac 3 8 n. 
			\label{eq:hankelhelp2}
		\end{equation}
		By Wedin's sine-theta theorem \cite[Chapter V, Theorem 4.4]{stewart1990matrix}, inequality \eqref{eq:hankelhelp1}, \eqref{eq:hankelhelp2} and that $C\alpha \leq 1/400$, 
		$$
		\norm{\sin( U, U_\psi)}_2
		\leq \frac{\| E\|_2}{\sigma_r(H(\bfy))-\|E\|_2}
		< \frac{C \alpha n}{\sigma_r(H(\bfy))-\| E\|_2}
		\leq 4C  \alpha
		= \frac 1 {100} \alpha.
		$$
		
		Let us relate this result to the Fourier subspace problem described in \cref{sec:Fouriersubspace}. For any $(\bfx,\bfa)\in\calP(16\pi/n,r)$, given data of the form $\Phi_n(\bfx)\bfa+\bfz\in \C^n$ and arbitrary $\bfz$ such that $\|\bfz\|_2\leq \alpha\sqrt n$, we have produced a method $U_\psi$ such that estimates the Fourier subspace $U:=U(\bfx)$ better than the minimax lower bound in \cref{thm:Fouriersubspace2}. This is a contradiction, which completes the proof.  	
	\end{proof}
	
	Now we are ready to derive a two-sided bound for the minimax error which will also show that Gradient-MUSIC produces an exponential sum and its performance guarantee is optimal in $n$, the noise-to-signal ratio $\alpha^2$, and subspace dimension $r$.
	
	\begin{corollary}\label{cor:gradmusicexponential}
		There is an absolute constant $C>0$ such that the following hold. For any $r\geq 1$, $n\geq \max\{200,16r\}$, and sufficiently small $\alpha \leq 1$, we have
		\begin{equation*}
				\frac 1 {400} \alpha \sqrt n 
				\leq E_*\round{n,r,\alpha \sqrt{n}}
				\leq C\alpha \sqrt n. 
		\end{equation*}
	\end{corollary}
	
	\begin{proof}
		Pick $\alpha\leq 1$ sufficiently small so that the assumptions of \cref{thm:gradmusicexponentialsum} hold. The lower bound for $E_*$ follows immediately from \cref{thm:minimaxexponential}. For the upper bound, since Gradient-MUSIC is a method for this problem, by the second inequality of \cref{thm:gradmusicexponentialsum}, we get 
		$$
		E_*\round{n,r,\alpha \sqrt n}
		\leq \sup_{(\bfx,\bfa)\in \calP(16\pi/n,r)} \sup_{\|\bfz\|_2\leq \alpha \sqrt n} \, \norm{f_\sharp - f}_{\ell^2_n}
		\leq C\alpha \sqrt n. 
		$$	
		This finishes the proof. 
	\end{proof}
	
	\subsection{Discussion}
	\label{sec:discussion3}
	
	An obstruction of exponential sum approximation is that in the general case where $\bfx$ is arbitrary, the exponential sum $f$ is nonharmonic (e.g., $\bfx$ could be irrational). In particular, $f$ cannot be interpreted as a Fourier series, which would be case under an additional requirement that $\bfx$ lies on a grid of width $1/N$ for an integer $N$. This also precludes the use of sparse recovery methods that are typically encountered in the compressed sensing literature. Another obstruction is the additional requirement that both $f$ and an approximant $\tilde f$ consist of exactly $r$ exponentials. This is a sparsity constraint, and it makes the set of exponential sums with $r$ terms a nonlinear set. 
	
	It may be helpful to compare \cref{thm:gradmusicexponentialsum} with the classical Whittaker–Nyquist–Shannon sampling formula for bandlimited functions \cite[Theorem 3.10.1]{christensen2003introduction}. It states that any function $g\in L^2(\R)$ bandlimited to $[-\pi,\pi]$ (i.e., the Fourier transform of $g$ with the $e^{-i\xi t}$ convention vanishes outside $[-\pi,\pi]$) can be uniquely represented as a $L^2(\R)$ convergent series, 
	$$
	g(t)=\sum_{j\in \Z} g(j) \, {\rm sinc}(t-j), \wherespace {\rm sinc}(t):=\frac{\sin(\pi t)}{\pi t}.
	$$
	This is also an interpolation formula and is a linear map of the samples $g|_{\Z}$ to a bandlimited function in $L^2(\R)$. It is stable to perturbations: if $g|_\Z+z$ is used in the reconstruction formula for a $z\in \ell^2(\Z)$, then the $L^2(\R)$ approximation error is
	\begin{equation}\label{eq:shannonerror}
	\Bigg\| g - \sum_{j\in \Z} (g(j)+z_j) \, {\rm sinc}(\cdot-j) \Bigg\|_{L^2(\R)}
	= \Bigg\| \sum_{j\in \Z} z_j \, {\rm sinc}(\cdot-j) \Bigg\|_{L^2(\R)}
	=\|z\|_{\ell^2(\Z)},
	\end{equation}
	which is a consequence of $\{{\rm sinc}(\cdot-j)\}_{j\in\Z}$ being an orthonormal sequence in $L^2(\R)$. 
	
	There are similarities and differences between the aforementioned stability result \eqref{eq:shannonerror} and \cref{thm:gradmusicexponentialsum}, even though the sampling theorem cannot be applied to exponential sum $f\in \calE(h,r)$ because $f\not\in L^2(\R)$. Nonetheless, $f$ can be thought of as the Fourier transform of a discrete measure $\sum_{k=1}^r a_k \delta_{x_k}$ supported in $[-\pi,\pi]$. Due to the restriction that $f$ has exactly $r$ exponentials (hence $2r$ unknowns), $f$ can be recovered from $n$ noiseless samples whenever $n\geq 2r$ (a result which can be traced back to Prony \cite{prony1795essai}), rather than needing its samples on the lattice $\Z$. The main difference is that the set of exponential sums with $r$ terms is not a linear space, while the space of bandlimited functions is a linear subspace of $L^2(\R)$. Gradient-MUSIC provides a computationally efficient and nonlinear method for computing an approximate exponential sum $f_\sharp$. While \cref{thm:gradmusicexponentialsum} can be thought of as a nonlinear analogue of the Whittaker–Nyquist–Shannon sampling theorem, we have pointed out several major differences. 
	
	Other works also used ideas from spectral estimation for function approximation, such as \cite{beylkin2005approximation,plonka2019computation}. Those papers are considerably different from this one in two aspects. They do not provide performance guarantees in the presence of noise and they use exponential sum approximation as a subroutine for more complicated approximation problems. It would be interesting to determine how the results of this paper can be used in their settings. 

	\section{Numerical simulations}
	\label{sec:numerics}
	
	We compare the Gradient-MUSIC Toeplitz estimator with alternating projection, where the latter is a standard computational method for the matrix problem. See \cite{fannjiang2025optimality} for further details about Gradient-MUSIC and \cite{chu2003structured} for how alternating projection works. An implementation of the Gradient-MUSIC Toeplitz estimator and code that reproduces the results in this paper can be found on WL's Github repository.\footnote{\href{https://github.com/weilinlimath}{https://github.com/weilinlimath}} 
	
	To set up this experiment, for given matrix size $n$ and rank parameter $r$, we need to create various $ T$ and $ E$. We first need to select $\bfx$ of cardinality $r$ with $\Delta(\bfx)\geq 2\pi \beta/n$ for a chosen parameter $\beta>1$. To do this, we draw independent $\gamma_1,\dots,\gamma_r$ from the uniform distribution on $[-\beta/2,\beta/2]$ and consider the set 
	$$
	\bfx=\left\{ \frac{4\pi \beta j +2\pi \gamma_j}m \right\}_{j=1}^r.
	$$ 
	That is, our quasi-random set $\bfx$ is a perturbation of $r$ consecutive points separated by $2\pi(2\beta)$. Note that $\Delta(\bfx)\geq 2\pi\beta/m$ as desired, where it is possible for equality to hold. We chose a quasi-random method to generate $\bfx$ to avoid special properties that could potentially have special interactions with the Fourier transform.  
	
	For the amplitudes $\bfa$, we let $a_1,\dots,a_r$ be i.i.d. $\{\pm 1\}$ Rademacher random variables. Thus, we have specified some way of generating 
	$$
	(\bfx,\bfa)\in \calP\left( \frac{2\pi \beta} m, 1\right). 
	$$
	For such a pair, we set 
	$$ 
	T:=\Phi_n(\bfx)\diag(\bfa)\Phi_n(\bfx)^*.$$ 
	If $\beta\geq 4$, then
	$ T\in \calT(n,r).$ We introduced $\beta$ as a general parameter because we will perform some experiments where $\beta=2$, which lies outside the theory for Gradient-MUSIC. 
	
	To form $ E$, we let  $e_{-n+1},\dots,e_{n-1}$ be i.i.d. complex normal random variables with mean zero and variance $\sigma^2$. Let $ E$ be the Toeplitz matrix where $E_{j,k}=e_{j-k}$. As explained earlier in \cref{sec:toeplitz}, with probability $1-o(m)$, it holds that $\| E\|_2\lesssim \sigma\sqrt{n\log(n)}$, hence $\| E\|\leq \alpha n$ for large enough $n$ depending on $\alpha$ and $\sigma$. Also explained there, under favorable scenarios, $r$ can be correctly determined by singular value thresholding, so we assume that $r$ is known in our simulations. 
	
	Both the Gradient-MUSIC Toeplitz estimator and alternating projection need the correct rank when they compute the truncated singular value decomposition. The former only performs one truncated SVD up front, while the latter does this for each iteration. Also for alternating projection, we use 50 iterations or if the iterates stall (specifically, if the Frobenius norm between consecutive iterates are less than $10^{-4}\| T+ E\|_{\rm F}$). However, in our simulations, the second termination condition is rarely ever fulfilled, which is consistent with prior observations whereby alternating projection is seen to converge slowly due to ``ping pong" effects.  
	
	For each set of parameters $(n,r,\beta,\sigma)$, we repeat this experiment over 10 trials where $(\bfx,\bfa)$ and $ E$ are generated from the previously described process. We record two primary statistics: relative error and computational time in seconds. If $\tilde T$ is produced by some method, the relative error is defined as 
	$$
	\frac{\|\tilde T- T\|_{\rm F}}{\| T\|_{\rm F}}.
	$$
	For each statistic, we report both the average and max over these trials. All computations were done on a commercial machine with Apple M2 chip and 16 GB RAM. Since alternating projection does not always produce a rank $r$ Toeplitz matrix, we also record the median and max $\epsilon$-rank of the computed Toeplitz matrix. The $\epsilon$-rank of $ A$ is the number of singular values which are at least $\epsilon$ of the maximum singular value, namely, $$|\{k\colon \sigma_k( A)\geq \epsilon \sigma_1( A)\}|.$$
	
	\begin{table}[ht]
		\centering
		\begin{tabular}{|c|c|c|c|} \hline
			{\bf Gradient-MUSIC} &$(n,r)=(200,20)$ &$(n,r)=(500,50)$ &$(n,r)=(1000,100)$\\ \hline 
			average error  &0.00702 &0.00475 &0.00296 \\
			max error &0.00841 &0.00991 &0.00311 \\
			average time (sec) &0.0817 &0.472 &1.89 \\
			max time (sec) &0.122 &0.922 &2.24  \\ \hline 
			{\bf Alternating projection} &$(n,r)=(200,20)$ &$(n,r)=(500,50)$ &$(n,r)=(1000,100)$ \\ \hline 
			average error &0.00783 &0.00466 &0.00333 \\
			max error &0.00918 &0.00522 &0.00346\\
			average time (sec) &1.23 &12.1 &163 \\
			max time (sec) &1.36 &16.2 &179 \\ 
			median $10^{-6}$-rank &198 &479 &880 \\
			max $10^{-6}$-rank &200 &484 &893 \\ \hline 
		\end{tabular}
		\caption{Results of Experiment 1 where $\beta=4$ and $\sigma=0.1$.}
		\label{tab:exper1}
	\end{table}
	
	For Experiment 1, we fix $\beta=4$ and $\sigma=0.1$. The theory developed for the Gradient-MUSIC Toeplitz estimator in \cref{sec:toeplitz} is applicable. The result of this first experiment for variable $(n,r)$ is shown in \cref{tab:exper1}. We observed that Gradient-MUSIC and alternating projection have similar errors, except the latter is significantly more expensive. This behavior is expected because the main computational bottleneck in Gradient MUSIC is a (single) truncated SVD step (which corresponds to the $n^2r$ term in complexity cost \eqref{eq:computationalcost1}), whereas alternating projection uses a truncated SVD in each step of the iteration. That is, the complexity of Gradient-MUSIC Toeplitz estimator is comparable to just one iteration of alternating projection. Additionally, we see that alternating projection does not output the correct rank since the $10^{-6}$-rank is close to $n$. For those interested, we record the $10^{-2}$-rank in the subsequent experiment.   
	
	\begin{table}[ht]
		\centering
		\begin{tabular}{|c|c|c|c|} \hline
			{\bf Gradient-MUSIC} &$(n,r)=(200,40)$ &$(n,r)=(500,100)$ &$(n,r)=(1000,200)$\\ \hline 
			average error  &0.0686 &0.0452 &0.0346 \\
			max error &0.0728 &0.0643 &0.0516 \\
			average time (sec) &0.0545 &0.472 &2.66 \\
			max time (sec) &0.121 &0.716 &3.64  \\ \hline 
			{\bf Alternating projection} &$(n,r)=(200,40)$ &$(n,r)=(500,100)$ &$(n,r)=(1000,200)$ \\ \hline 
			average error &0.0766 &0.0470 &0.0338 \\
			max error &0.0808 &0.0507 &0.0346 \\
			average time (sec) &1.80 &23.3 &214 \\
			max time (sec) &1.96 &28.1 &230 \\ 
			median $10^{-2}$-rank &55.5 &102 &200 \\
			max $10^{-2}$-rank &61 &114 &201 \\ \hline 
		\end{tabular}
		\caption{Results of Experiment 2 where $\beta=2$ and $\sigma=1$.}
		\label{tab:exper2}
	\end{table}
	
	For Experiment 2, we pick $\beta=2$ and $\sigma=1$. This is a more difficult setup than Experiment 1 since we expect the Toeplitz approximation problem to be more sensitive to noise as $\beta$ approaches 1. Making $\beta$ smaller also allows us to make $|\bfx|$ (hence $r$) larger, since $\Delta\geq 2\pi \beta/m$ implies that $|\bfx|\leq m/\beta$. Additionally, Experiment 2 uses larger $\sigma$. We also record the $10^{-2}$-rank for alternating projection. The results of this second experiment for variable $(n,r)$ is shown in \cref{tab:exper2}. Again, we observed that Gradient-MUSIC and alternating projection have similar performance guarantees, but the former is significantly faster. Alternating projection does not output the correct rank, even if we examine the $10^{-2}$-rank. 
	
	We do not compare with structured total least norm \cite{rosen1996total,park1999low} due to its high computational cost. We also do not compare with \cite{musco2024sublinear} since their method is not easy to implement, and we were unable to find an implementation of their method on the internet.

	\section{Prospect and Discussion}
	\label{sec:conclusion}

This paper connects four themes: spectral estimation, Fourier subspace estimation, structured low-rank Toeplitz/Hankel approximation, and approximation of exponential sums.  The common algebraic object is the Fourier/Vandermonde matrix $
        \Phi_n(\bfx)=\bigl[e^{ijx_k}\bigr],
$
and the common analytical issue is the stability of this matrix under finite sampling.  In the
well-separated regime, where $\Delta(\bfx)\gtrsim n^{-1}$, this stability allows one to move
quantitatively between samples, subspaces, structured matrices, and frequencies.

The proofs also show that this role is not merely algorithmic.  The lower bounds for the other three problems are obtained by
connecting them to spectral estimation.  If any of the three  problems could be solved with a substantially better
rate than the minimax rate, then combining such a method with Gradient-MUSIC would yield a spectral estimation
algorithm that violates the minimax lower bound for spectral estimation.  This gives a quantitative transference principle:
in the separated regime, the difficulty of spectral estimation propagates to the associated
structured approximation problems.

\paragraph{Are the problems equivalent (in the well-separated case)?}

The results in this paper show that, in the well-separated case, spectral estimation provides
one direction of equivalence: together with standard subspace perturbation theory, it yields
optimal algorithms for the other three.  The minimax lower bounds show conversely that neither of these structured
approximation problems can be substantially easier than spectral estimation. 

The local bi-Lipschitz equivalence in \cref{prop:bilipschitz} proves that spectral and Fourier subspace estimation are equivalent in the setting of this paper. In a sufficiently small neighborhood of a well-separated $\bfx$, the local chart
\[
        \bfx\mapsto U(\bfx):=\range[\Phi_n(\bfx)]\subset\C^n
\]
is a well-conditioned coordinate chart for the sub-manifold of Fourier subspaces embedded in the Grassmannian, with
\begin{eqnarray}
        \|\sin(U(\bfx),U(\widetilde{\bfx}))\|_2
      &  \asymp&
        n\|\bfx-\widetilde{\bfx}\|_\infty\\
              \|\sin(U(\bfx),U(\widetilde{\bfx}))\|_{\rm F}
      &  \asymp&
        n\|\bfx-\widetilde{\bfx}\|_2. 
\end{eqnarray}

Then locally, frequency estimation and Fourier subspace estimation
are the same inverse problem written in two different coordinates, with the natural rescaling of
the loss by a factor of \(n\). Consequently, an optimal frequency estimator $\tilde\bfx$ induces an optimal Fourier subspace
estimator by 
$$
U(\tilde\bfx):=\range[\Phi_n(\tilde\bfx)].
$$
On the other hand, an optimal Fourier subspace
estimator \(\tilde U\) induces an optimal frequency estimator by applying a local inverse of the chart \( \bfx\mapsto U(\bfx)\).

For structured low-rank
Toeplitz/Hankel approximation and exponential sums, the relation to frequency estimation is less direct because
one must additionally estimate the amplitudes. 

\paragraph{The role of the truncated SVD.}

Our algorithms for the three structured approximation problems begin with a truncated singular value decomposition. For the Fourier subspace and exponential sum problems, the leading singular subspace of the Hankel lift
$H(y+z)$ approximates the Fourier subspace associated with the samples.  For the Toeplitz problem,
the leading singular subspace of $M=T+E$ approximates the Fourier subspace associated with
$T$.  

On one hand, the Eckart-Young-Mirsky theorem already implies the truncated SVD is an optimal low-rank approximation of a general matrix in both the spectral and Frobenius norms. On the other hand, structured approximation problems also incorporate additional constraints, so it is not clear if this first approximation procedure is even desirable. In the regime studied here, in all three approximation problems, the truncated SVD preserves all information needed for optimal structured recovery and suggests that it is not merely a heuristic denoising step.  

This observation may be useful beyond the present setting.  Many algorithms for inverse problems
already begin with a low-rank approximation or empirical subspace estimate.  The results here
give a setting in which such a reduction can be justified sharply: after the SVD, the remaining
task is to impose the correct nonlinear structure.

\paragraph{Beyond the well-separated regime.}

The present theory is deliberately focused on the well-separated regime.  This is the regime in
which Fourier/Vandermonde matrices are uniformly well-conditioned and there exist local stability results. Relaxing the separation condition is an important and
substantially harder direction.  When frequencies cluster below the Rayleigh length, which is on the order of $1/n$, the sampling
map remains identifiable in principle but becomes severely ill-conditioned.  In that regime, the
manifold of Fourier subspaces may develop highly anisotropic geometry, and the minimax rates
should depend on the cluster geometry rather than only on $n$ and the noise level.

\paragraph{Conclusion.}

The main contribution of this paper is a quantitative transference principle centered on spectral
estimation.  Sampling and interpolation by exponential sums generate low-rank Hankel and
Toeplitz structures; these structures determine Fourier subspaces; and spectral estimation converts
approximate subspaces into accurate frequencies.  Once the frequencies are recovered, both
structured low-rank matrix approximation and Fourier subspace estimation follow with optimal
rates.

	\section*{Acknowledgments}
	
	Both authors thank Wenjing Liao for insightful and valuable discussions. We also thank anonymous reviewers for their comments and suggestions which have greatly improved the scope and presentation of this manuscript. We are especially grateful for a reviewer's generous suggestion of the proof of \cref{lem:Phidifferencespectralnorm}, which led to improvements in \cref{thm:Fouriersubspace,thm:gdmusictoeplitz} of the first version of this paper. WL is partially supported by NSF-DMS Award \#2309602 and a PSC-CUNY Cycle 56 Grant.

	\bibliographystyle{plain}
	\bibliography{Toeplitzbib}

	\appendix

	\section{Proof of \cref{thm:maingradientMUSIC2}}
	\label{proof:maingradientMUSIC2}
	\begin{proof}
		Recall the two MUSIC functions defined in \eqref{eq:musicfunctions}. Since
		$$
		\|\sin(U, W)\|_2
		\leq \norm{\sin(U, W)}_{\rm F}
		\leq \frac 1 {100},
		$$
		the conclusions of \cref{thm:maingradientMUSIC} and the discussion thereafter hold. In view of those properties, let $\tilde\bfx$ be the $r$ smallest local minima of $\tilde q$ indexed to best match $\bfx$. 
		
		We first claim that the second part of the theorem follows from \eqref{eq:goal0}. Indeed, gradient descent converges linearly to $\tilde\bfx$ and using enough iterates produces $\bfx^\sharp$ such that
		$$
		\|\bfx^\sharp-\tilde \bfx\|_2
		\leq \frac {5} n \, \norm{\sin( U, W)}_{\rm F}. 
		$$
		Then by triangle inequality, we see that
		$$
		\|\bfx^\sharp-\bfx\|_2 
		\leq \|\bfx^\sharp-\tilde \bfx\|_2 + \|\tilde\bfx- \bfx\|_2
		\leq \frac {25} n \norm{\sin( U, W)}_{\rm F},
		$$
		which yields the theorem.
		
		To prove \eqref{eq:goal0}, we will invoke a multidimensional version of the mean value theorem, which requires some preparation. For convenience, we set $\varTheta:=\norm{\sin(U, W)}_{\rm F}$. Let $\bft=(t_1,\dots,t_r)\in \T^r$, and we view $\bfx=\{x_k\}_{k=1}^r$ as a coordinate $\bfx=(x_1,\dots,x_r)\in\T^r$. Define the functions $Q,\tilde Q\colon \T^r\to \R$ by
		\begin{align*}
			Q(\bft)&=q(t_1)+\cdots+q(t_r), \\
			\tilde Q(\bft)&=\tilde q(t_1)+\cdots+\tilde q(t_r).
		\end{align*}
		Note that each $x_k$ is a double root of $q$ and consequently,  
		$
		Q(\bfx)=0$  and $\nabla Q(\bfx)=\bfzero.$ Since each $\tilde x_k$ is a critical point of $\tilde Q$, we also have $\nabla Q(\tilde\bfx)=\bfzero$. The goal is to show that
		\begin{equation}\label{eq:goal}
			\nabla \tilde Q(\bft) \cdot (\bft-\bfx)\geq 0, \forallspace \|\bft-\bfx\|_2=\frac{20\varTheta} n
		\end{equation} 
		If this goal can be accomplished, then a higher dimensional intermediate value theorem \cite[Theorem 1]{morales2002bolzano} ensures that $\tilde Q$ has a critical point $\bfu=(u_1,\dots,u_r)$ such that
		\begin{equation}
			\label{eq:goalball}
			\|\bfu-\bfx\|_2\leq \frac{20\varTheta} n. 
		\end{equation}
		Now we argue that necessarily $\bfu=\tilde \bfx$. The discussion following \cref{thm:maingradientMUSIC} asserts that $\tilde x_k$ is the only critical point of $\tilde q$ in the interval $[\tilde x_k-4\pi/(3n),\tilde x_k + 4\pi/(3n)]$. By \eqref{eq:minimaperturbation}, \eqref{eq:goalball}, and the assumption $\varTheta\leq 0.01$, we have
		$$
		\max_k |\tilde u_k- \tilde x_k|
		\leq \|\tilde \bfu- \bfx\|_2 + \max_k |x_k- \tilde x_k|
		\leq \frac{27\varTheta}{n}
		\leq \frac{0.27}{n}. 
		$$
		This proves that $\bfu=\tilde\bfx$.
		
		The rest of the proof shows that \eqref{eq:goal} holds. Let $\alpha\in [1,5]$ and we will give a parameter tuning argument based on $\alpha$. Let $\bft\in \T^r$ such that  $\|\bft-\bfx\|_2=2\pi \alpha \varTheta/n$. By Taylor's theorem applied to $q'$ around $x_k$ and that $q'(x_k)=0$, there is a $\xi_k$ between $x_k$ and $t_k\in \T$ such that 
		$$
		q'(t_k)=q''(x_k)(t_k-x_k)+\frac 12 q'''(\xi_k)(t_k-x_k)^2. 		
		$$
		This yields
		\begin{align*}
			\nabla Q(\bft)\cdot (\bft-\bfx)
			&=\sum_{k=1}^r q''(x_k) (t_k-x_k)^2 + \frac 12 \sum_{k=1}^r q'''(\xi_k) (t_k-x_k)^3 \\
			&\geq \left(\min_{k=1,\dots,r} q''(x_k) \right) \|\bft-\bfx\|_2^2 -  \left(\max_{k=1,\dots,r} \frac 12 \, |q'''(\xi_k)||t_k-x_k| \right) \|\bft-\bfx\|_2^2. 
		\end{align*}
		By Bernstein's inequality applied to the trigonometric polynomial $(\tilde q-q)'$ which has degree $n$ and that the $k$-th column of $\Phi_n(\bft)$ is $\sqrt n \, \bfphi(t_k)$, we get
		\begin{align*}
			|\tilde q \, '(t_k)-q'(t_k)|
			&\leq n \, |\tilde q (t_k)-q(t_k)| \\
			&=n \, |\bfphi(t_k)^* (P_W-P_U)\bfphi(t_k)| 
			\leq n \|(P_W-P_U) \bfphi(t_k)\|_2.  
		\end{align*}
		Using this inequality and the relationship between $P_W-P_U$ and Frobenius distance \eqref{eq:Chordal}, we see that 
		\begin{align*}
			\Big\|\nabla \tilde Q(\bft)-\nabla Q(\bft)\Big\|^2_2
			&= \sum_{k=1}^r |\tilde q \, '(t_k)-q'(t_k)|^2 \\
			&\leq n^2 \sum_{k=1}^r \|(P_W-P_U) \bfphi(t_k)\|_2^2 
			=n \left\|(P_W-P_U)\Phi_n(\bft) \right\|_{\rm F}^2 
			\leq 2n \varTheta^2 \|\Phi_n(\bft)\|_2^2.  
		\end{align*}
		Since $\Delta(\bfx)\geq 8\pi/n$ and $2\pi\alpha\varTheta\leq 2\pi(0.05)$, we have $\Delta(\bft)\geq (8\pi-2\pi(0.1))/n$. By \eqref{eq:Fouriersigs}, we have
		$$
		\|\Phi_n(\bft)\|_2
		\leq C_1\sqrt n, \wherespace C_1:= \sqrt{1+(3.9)^{-1}} .
		$$
		Using the above inequalities, we see that 
		\begin{equation} \label{eq:Qlowerhelp1} 
			\begin{split}
			\nabla \tilde Q(\bft)\cdot (\bft-\bfx)
			&\geq \nabla Q(\bft)\cdot (\bft-\bfx) - \left\|\nabla Q(\bft)- \nabla \tilde Q(\bft)\right\|_2 \|\bft-\bfx\|_2 \\
			&\geq \frac 1 {n^2}\left(\min_{k=1,\dots,r} q''(x_k) - \max_{k=1,\dots,r} \frac 12 \, |q'''(\xi_k)||t_k-x_k| \right) 4\pi^2\alpha^2 \varTheta^2  - 2 \pi \alpha 2 C_1 \varTheta^2.
		\end{split}
		\end{equation}
		We will argue that $\alpha=20/(2\pi)$ ensures that the right hand side of \eqref{eq:Qlowerhelp1}  is nonnegative, thereby proving \eqref{eq:goal}. 
		
		To get explicit control over terms involving the second and third derivatives of $q$ in inequality \eqref{eq:Qlowerhelp1}, the main machinery developed in \cite[Section 9]{fannjiang2025optimality} is an approximation of $q$ (and its derivatives) by shifts of the normalized Fej\'er kernel (and its derivatives), denoted
		$$
		f_n(t) := \frac 1 {n^2} \left( \frac{\sin(nt/2)}{\sin(t/2)}\right)^2. 
		$$
		The approximation error between $q^{(\ell)}(t)$ by $-f_n^{(\ell)}(t-x_k)$ for all $t$ sufficiently close to $x_k$ involve various explicit expressions and/or integrals. From here onward, set $T_\ell:=T_\ell(100,4-\alpha\varTheta,4)$ for $\ell\in \{0,1,2,3\}$ and $C_0:= C_0(100,4)$, where all of these quantities are explicitly defined in \cite[Section 9]{fannjiang2025optimality} and can be numerically evaluated to arbitrary precision. Then the results there show that
		\begin{align*}
			\min_k q''(x_k) 
			&\geq C_0 n^2, \\
			|q'''(t_k)+f_n'''(t_k-x_k)|
			&\leq \big(2T_0T_3+6T_1T_2 \big)n^3. 
		\end{align*}
		Note that $\|f'''_n\|_{L^\infty(\T)}\leq n^3/10$, which follows from direct computation using the Fourier series representation of $f_n$. Also recall that $\xi_k$ is between $x_k$ and $t_k$ and $|x_k-t_k|\leq \|\bfx-\bft\|_2=2\pi \alpha \varTheta/n$. Then we get that
		\begin{align*}
			\frac 12 |q'''(\xi_k)| |t_k-x_k|
			&\leq \frac 12 \Big(|f_n'''(\xi_k-x_k)| + \big(2T_0T_3+6T_1T_2 \big)n^3 \Big) |t_k-x_k| \\
			&\leq\left( \frac 1 {20}+  T_0T_3+3T_1T_2 \right)2\pi \alpha\varTheta n^2 . 
		\end{align*}
		Then inserting this into \eqref{eq:Qlowerhelp1}, we end up with the lower bound
		\begin{align*}
			\nabla \tilde Q(\bft)\cdot (\bft-\bfx)
			&\geq \left( C_0 - \left( \frac 1 {20}+  T_0T_3+3T_1T_2 \right)2\pi \alpha\varTheta  \right) 4\pi^2\alpha^2 \varTheta^2  - 2 \pi \alpha 2C_1 \varTheta^2.
		\end{align*}
		Finally, we set $\alpha=20/(2\pi)$ and numerically evaluate $T_0,T_1,T_2,T_3,C_0,C_1$ to see the right hand side is positive. This finishes the proof of \eqref{eq:goal} and consequently, the theorem.
	\end{proof}

	\section{Stability estimates} 
	
	\subsection{Local stability of Fourier matrices}

	\begin{lemma}
		\label{lem:Phidifferencespectralnorm}
		For $n\geq 2r$ and $\beta \geq 2$, suppose $\bfx,\tilde\bfx\subset \T$ are both sets of cardinality $r$ such that $\Delta(\bfx)\geq 2 \pi \beta/n$ and $\|\bfx-\tilde\bfx\|_\infty \leq \pi \beta/(2n)$. Then 
		\begin{align*}
			\big\|\Phi_n(\bfx)-\Phi_n(\tilde\bfx) \big\|_2
			&\leq \sqrt{\frac{\beta + 2}{\beta}} \, n^{3/2} \,  \|\bfx-\tilde\bfx\|_\infty, \\
			\big\|\Phi_n(\bfx)-\Phi_n(\tilde\bfx) \big\|_{\rm F}
			&\leq \sqrt{\frac{\beta + 2}{\beta}} \, n^{3/2} \,  \|\bfx-\tilde\bfx\|_2. 
		\end{align*}
	\end{lemma}
	\begin{proof}		
		Suppose $\bfx$ and $\tilde\bfx$ have been indexed so that their matching distance is minimized with respect to this ordering. For any $t\in [0,1]$ and $k=1,\dots,r$, define 
		$$
		x_{t,k}:=x_k + t \, ( \tilde x_k- x_k), \quad 
		\bfx_t := \{x_{t,k}\}_{k=1}^r, \andspace
		D := \diag(0,i,\dots,i(n-1)).
		$$
		A calculation shows that 
		$$
		\bfphi'(x_{t,k}) = D \bfphi(x_{t,k}).
		$$
		Importantly, this formula holds for each $t$ and $k$, with the same matrix $ D$. By the fundamental theorem of calculus,
		$$
		\bfphi( \tilde x_k)-\bfphi(x_k)
		= \int_0^1 \frac d {dt} (\bfphi(x_{t,k})) \, dt
		= \int_0^1  D \bfphi(x_{t,k}) ( \tilde x_k-x_k)\, dt.	
		$$
		This formula now implies 
		\begin{equation}
			\label{eq:intformula}
			\Phi_n(\tilde\bfx)-\Phi_n(\bfx)
			= \int_0^1  D\Phi_n(\bfx_{t}) \diag(\tilde\bfx-\bfx) \, dt.
		\end{equation}
		By a direct computation, we see that 
		$$
		\Delta(\bfx_t)
		\geq \Delta(\bfx)- 2t \|\tilde\bfx-\bfx\|_\infty
		\geq \frac{2\pi \beta}{n} - \frac{\pi \beta}{n}
		= \frac{\pi \beta}{n}.
		$$
		Note this lower bound is uniform in $t$ and \eqref{eq:Fouriersigs} implies
		$$
		\|\Phi_n(\bfx_t)\|_2
		\leq \sqrt{(1+2\beta^{-1}) n}. 
		$$
		This, together with \eqref{eq:intformula} and that $\| D\|_2\leq n$, shows that
		\begin{align*}
			\big\|\Phi_n(\tilde\bfx)-\Phi_n(\bfx) \big\|_2
			&\leq \| D\|_2 \left(\sup_{t\in [0,1]} \|\Phi_n(\bfx_{t})\|_2 \right) \|\diag(\tilde\bfx-\bfx)\|_2
			\leq \sqrt{\frac{\beta+2}{\beta}} \, n^{3/2} \,  \|\tilde\bfx-\bfx\|_\infty, \\
			\big\|\Phi_n(\tilde\bfx)-\Phi_n(\bfx) \big\|_{\rm F}
			&\leq \| D\|_2 \left(\sup_{t\in [0,1]} \|\Phi_n(\bfx_{t})\|_2 \right) \|\diag(\tilde\bfx-\bfx)\|_{\rm F}
			\leq \sqrt{\frac{\beta+2}{\beta}} \, n^{3/2} \,  \|\tilde\bfx-\bfx\|_2. 
		\end{align*}
	\end{proof}
	
	\subsection{Stability of quadratic method for amplitudes}
	The next result deals with reconstruction of amplitudes using the following ``quadratic" method. Say $A:=\diag(\bfa)\in \C^{r\times r}$, $\Phi:=\Phi_n(\bfx)$, and $\tilde \bfx$ is a given set that approximates $\bfx$. Let $E\in \C^{n\times n}$ be arbitrary. Consider the quantities 
	\begin{equation}
		\label{eq:quadratic}
		\tilde \Phi:= \Phi_n(\tilde \bfx), \quad 
		\tilde A := \tilde\Phi^\dagger  \big(\Phi A \Phi^* + E \big) \tilde\Phi ^{\dagger,*}, \andspace
		\tilde \bfa := \diag( \tilde A).
	\end{equation}
	Note that any permutation on $\tilde \bfx$ induces the same re-ordering of $\tilde \bfa$. Hence we always index $\tilde\bfx$ to best match $\bfx$, and that same ordering is used to index $\tilde\bfa$.

	\begin{lemma}
		\label{lem:amplitudes}
		For $n\geq 2r$ and $\beta \geq 2$, suppose $\bfx$ and $\tilde\bfx$ are sets of cardinality $r$ such that $\Delta(\bfx)\geq 2 \pi \beta/n$ and $\|\tilde\bfx-\bfx\|_\infty \leq \pi \beta/(4n)$. For any matrix $E\in\C^{n\times n}$, the vector $\tilde \bfa$ defined in \eqref{eq:quadratic} satisfies 
		\begin{align*}
			\|\bfa - \tilde \bfa\|_\infty
			&\lesssim n \, \|\bfa\|_\infty \|\bfx-\tilde\bfx\|_\infty + \frac {\|E\|_2} n,\\
			\|\bfa - \tilde\bfa\|_2
			&\lesssim n \, \|\bfa\|_\infty \|\bfx-\tilde\bfx\|_2 + \frac {\|E\|_{\rm F}} n.
		\end{align*}
	\end{lemma}	
	
	\begin{proof}
		Since
		\begin{align*}
			\|\bfa-\tilde \bfa\|_\infty
			&= \max_k \big|\bfe_k^* \big( A- \tilde A \big) \bfe_k \big|_2
			\leq \big\| A -\tilde A \big\|_2, \\
			\|\bfa-\tilde \bfa\|_2
			&= \left(\sum_{k=1}^r \big| A_{k,k}- \tilde A_{k,k} \big|_2^2 \right)^{1/2}
			\leq \big\| A -\tilde A \big\|_{\rm F},
		\end{align*}
		it suffices to control $A-\tilde A$ in both the spectral and Frobenius norms. We start with the formulas
		\begin{align*}
			\tilde\Phi^\dagger \Phi -  I 
			&= (\tilde\Phi^* \tilde\Phi)^{-1} \tilde\Phi^* (\Phi - \tilde\Phi),  \\
			\tilde A - A 
			&= \tilde\Phi^\dagger  \big(\Phi A \Phi^* + E \big) \tilde\Phi ^{\dagger,*}-A \\
			&=\big(\tilde\Phi^\dagger \Phi \big)  A \big(\tilde\Phi^\dagger \Phi -  I\big)^* + \big( \tilde\Phi^\dagger \Phi -  I\big)  A + \tilde\Phi^\dagger  E \tilde\Phi^{\dagger,*} \\
			&= \big(\tilde\Phi^\dagger \Phi \big)  A (\Phi - \tilde\Phi )^* \tilde\Phi (\tilde\Phi^* \tilde\Phi)^{-1}  +  (\tilde\Phi^* \tilde\Phi)^{-1} \tilde\Phi^* (\Phi - \tilde\Phi ) A + \tilde\Phi^\dagger  E \tilde\Phi^{\dagger,*} . 
		\end{align*}
		By our assumptions on $\bfx$ and $\tilde\bfx$, we have 
		$$
		\Delta(\tilde \bfx) \geq \Delta(\bfx)-2\|\bfx - \tilde\bfx\|_\infty
		\geq \frac{2\pi \beta} n - \frac{\pi \beta}{2n}
		\geq \frac{2\pi (3/2)}{n}.
		$$
		This allows us to use \eqref{eq:Fouriersigs} to see that 
		\begin{align*}
			\max\left\{\|\Phi \|_2, \, \|\tilde\Phi\|_2 \right\}
			&\lesssim {\sqrt n}, \andspace
			\max\left\{\|\Phi^\dagger \|_2, \, \|\tilde\Phi^\dagger\|_2 \right\}
			\lesssim \frac 1{\sqrt n}.
		\end{align*}
		Using the above results, we see that
		\begin{align*}
			\big\|\tilde A - A\big\|_2
			&\lesssim \frac 1 {\sqrt n} \, \|A\|_2 \,  \|\Phi-\tilde\Phi\|_2 + \frac {\|E\|_2} n, \\
			\big\|\tilde A - A \big\|_{\rm F}
			&\lesssim \frac 1 {\sqrt n} \, \|A\|_2 \, \|\Phi-\tilde\Phi\|_{\rm F} + \frac {\|E\|_{\rm F}} n.  
		\end{align*}
		Recalling that $\|A\|_2=\|\bfa\|_\infty$ and \cref{lem:Phidifferencespectralnorm} completes the proof.
	\end{proof}
	
	Formula \eqref{eq:quadratic} and \cref{lem:amplitudes} are written for the Toeplitz case since $\Phi A \Phi^*$ is a Toeplitz matrix. For the Hankel case, we replace $\Phi A\Phi^*$ with $\Phi A\Phi^T$ and $\tilde\Phi ^{\dagger,*}$ with $\tilde\Phi ^{\dagger,T}$ in definition \eqref{eq:quadratic}. Then \cref{lem:amplitudes} holds verbatim. 
	
	\subsection{Proof of \cref{prop:stability1}}
	\label{proof:stability1}
	
	\begin{proof}
		The proof is essentially a direct calculation. For convenience, let $\Phi:=\Phi_n(\bfx)$ and $\tilde \Phi:=\Phi_n(\tilde \bfx)$. From relationships \eqref{eq:Chordal} and \eqref{eq:sinetheta}, it suffices to control $P(\bfx)-P(\tilde \bfx)$. This can be expressed as
		\begin{align*}
			P(\bfx)-P(\tilde \bfx)
			&= \Phi(\Phi^* \Phi)^{-1} \Phi^* - \tilde \Phi(\tilde \Phi^* \tilde \Phi)^{-1} \tilde \Phi^* \\
			%			&= \Phi(\Phi^* \Phi)^{-1} \Phi^* - \tilde \Phi(\Phi^* \Phi)^{-1} \Phi^* + \tilde \Phi(\Phi^* \Phi)^{-1} \Phi^* - \tilde \Phi(\tilde \Phi^* \tilde \Phi)^{-1} \Phi^*  + \tilde \Phi(\tilde \Phi^* \tilde \Phi)^{-1} \Phi^* - \tilde \Phi(\tilde \Phi^* \tilde \Phi)^{-1} \tilde \Phi^* \\
			&= (\Phi- \tilde \Phi)(\Phi^* \Phi)^{-1} \Phi^* + \tilde \Phi(\Phi^* \Phi)^{-1} (\tilde \Phi^* \tilde\Phi - \Phi^* \Phi) (\tilde \Phi^* \tilde \Phi)^{-1} \Phi^* + \tilde \Phi(\tilde \Phi^* \tilde \Phi)^{-1} (\Phi-\tilde \Phi)^* \\
			&= (\Phi- \tilde \Phi)(\Phi^* \Phi)^{-1} \Phi^* + \tilde \Phi(\Phi^* \Phi)^{-1} \tilde \Phi^* (\tilde\Phi - \Phi) (\tilde \Phi^* \tilde \Phi)^{-1} \Phi^* \\
			&\quad + \tilde \Phi(\Phi^* \Phi)^{-1} (\tilde \Phi - \Phi)^* \Phi (\tilde \Phi^* \tilde \Phi)^{-1}  \Phi^* + \tilde \Phi(\tilde \Phi^* \tilde \Phi)^{-1} (\Phi-\tilde \Phi)^*. 
		\end{align*}
		Since $\Delta(\tilde\bfx)\geq \Delta(\bfx)-2\|\bfx-\tilde\bfx\|_\infty \geq 2\pi (3\beta/4)/n \geq 2\pi(3/2)/n$, we use \eqref{eq:Fouriersigs} to see that
		\begin{align*}
			\max\left\{ \|\Phi\|_2, \|\tilde \Phi\|_2\right\} &\lesssim \sqrt n, \\
			\max\left\{ \|(\Phi^* \Phi)^{-1}\|_2, \|(\tilde \Phi^* \tilde \Phi)^{-1}\|_2\right\} 
			&\lesssim \frac 1 n. 
		\end{align*}
		Thus, we have
		\begin{align*}
			\|P(\bfx)-P(\tilde \bfx)\|_2 
			&\lesssim \frac 1 {\sqrt n} \|\Phi- \tilde \Phi\|_2, \\
			\|P(\bfx)-P(\tilde \bfx)\|_{\rm F} 
			&\lesssim \frac 1 {\sqrt n} \|\Phi- \tilde \Phi\|_{\rm F}. 
		\end{align*}
		Additionally, $\Phi-\tilde \Phi$ can be controlled in both the spectral and Frobenius norms by \cref{lem:Phidifferencespectralnorm}, and note that $\sqrt{(\beta+2)/\beta}\leq \sqrt 2$ for all $\beta\geq 2$. 
	\end{proof}
	
	\subsection{Proof of \cref{prop:stability2}}
	\label{proof:stability2}
	
	\begin{proof}
		We first concentrate on the spectral norm estimate. We start with the calculation, 
		\begin{align*}
			T(\tilde \bfx,\tilde\bfa)- T(\bfx,\bfa)
			&=\tilde\Phi \, \tilde A \, \tilde\Phi^* - \Phi \, A \, \Phi^* \\
			&= \tilde\Phi\big( \tilde A - A \big) \tilde\Phi^*  + \tilde \Phi A \big(\tilde \Phi^* - \Phi^* \big) + \big( \tilde\Phi  - \Phi \big) A \,  \Phi^*. \\
		\end{align*}
		Then by triangle inequality, 
		\begin{align*}
			\norm{T(\tilde \bfx,\tilde\bfa)- T(\bfx,\bfa)}_2				
			&\leq \|\tilde\Phi \|_2  \big\| \tilde\bfa - \bfa\big\|_\infty \|\tilde\Phi \|_2 + \|\tilde\Phi \|_2 \|\bfa\|_\infty \big\| \tilde\Phi - \Phi \big\|_2 + \big\|  \tilde\Phi -\Phi \big\|_2  \|\bfa\|_\infty \|\Phi\|_2, \\
			\norm{T(\tilde \bfx,\tilde\bfa)- T(\bfx,\bfa)}_{\rm F}
			&\leq \|\tilde\Phi \|_2  \big\| \tilde\bfa - \bfa\big\|_2 \|\tilde\Phi \|_2 + \|\tilde\Phi \|_2 \|\bfa\|_\infty \big\| \tilde\Phi - \Phi \big\|_{\rm F} + \big\|  \tilde\Phi -\Phi \big\|_{\rm F}  \|\bfa\|_\infty \|\Phi\|_2. 
		\end{align*}
		To complete the proof, we use that $\|\bfa\|_\infty\leq 10$ by assumption,     together with \cref{lem:Phidifferencespectralnorm} (note that $\sqrt{(\beta+2)/\beta}\leq \sqrt 2$ for all $\beta\geq 2$), and \eqref{eq:Fouriersigs} to see that $\max\{\|\Phi\|_2, \|\tilde\Phi\|_2\}\lesssim \sqrt n$.
	\end{proof}
	
	\subsection{Proof of \cref{prop:stability3}}
	\label{proof:stability3}
	
	\begin{proof}
		We can re-index $\tilde \bfx$ such that the matching distance between $\bfx$ and $\tilde\bfx$ is minimized. For $s\in [0,1]$, define $x_{s,k}= x_k+s(\tilde x_k-x_k)$ and let $\bfx_s=\{x_{s,k}\}_{k=1}^r$. Note that uniformly in $s\in[0,1]$,
		\begin{equation}
			\label{eq:sephelp1}	
			\Delta(\bfx_s)
			\geq \Delta(\bfx)-2s\|\bfx-\tilde\bfx\|_\infty
			\geq \frac{2\pi \beta}{n}-\frac{\pi \beta}{n}
			\geq \frac{2\pi}n.
		\end{equation}
		By the fundamental theorem of calculus, we have
		\begin{align*}
			\tilde f(t) - f(t)
			&= \sum_{k=1}^r (\tilde a_k- a_k) e^{ i\tilde x_k t} + \sum_{k=1}^r a_k \left( e^{ i \tilde x_k t}-e^{ ix_k t}\right) \\
			&= \sum_{k=1}^r (\tilde a_k- a_k) e^{ i\tilde x_k t} + \sum_{k=1}^r a_k \int_0^1 e^{ ix_{s,k} t} \,  i t(\tilde x_k- x_k) \, ds \\
			&= \sum_{k=1}^r (\tilde a_k- a_k) e^{ i\tilde x_k t} + \int_0^1 t  \sum_{k=1}^r  i a_k (\tilde x_k- x_k) e^{ ix_{s,k} t} \, ds. 
		\end{align*}
		For the first term, we can directly use \eqref{eq:beurling2} (justified by \eqref{eq:sephelp1}) to see that
		$$
		\Bigg(\int_0^n \Big| \sum_{k=1}^r ( \tilde a_k- a_k) e^{ i\tilde x_k t} \Big|^2 \Bigg)^{1/2}
		\lesssim \sqrt n \, \|\tilde \bfa- \bfa\|_2.
		$$
		For the second term, for each $s\in [0,1]$, it is convenient to define the exponential sum 
		$$
		f_s(t):=\sum_{k=1}^r  i a_k (\tilde x_k- x_k) e^{ ix_{s,k} t}. 
		$$
		By the Minkowski integral inequality and \eqref{eq:beurling2} (which is justified by \eqref{eq:sephelp1}), 
		\begin{align*}
			\left(\int_0^n \Big| \int_0^1 t f_s(t)\, ds \Big|^2 \, dt\right)^{1/2}
			&\leq \int_0^1 \left( \int_0^n t^2 |f_s(t)|^2 \, dt \right)^{1/2} \, ds\\
			&\leq \int_0^1 n \left( \int_0^n |f_s(t)|^2 \, dt \right)^{1/2} \, ds \\
			&\lesssim n^{3/2}  \left(\sum_{k=1}^r |a_k(\tilde x_k-x_k)|^2 \right)^{1/2} \\
			&\lesssim n^{3/2} \|\bfa\|_\infty \, \|\tilde \bfx-\bfx\|_2.
		\end{align*}
		Combining the above completes the proof of 
		the first statement of the proposition. The second assertion of this proposition follows by the same argument except we use \eqref{eq:Fouriersigs} instead of \eqref{eq:beurling2}. 
	\end{proof}

\end{document}